\documentclass[a4paper,onecolumn,11pt,accepted=2021-01-16]{quantumarticle}
\pdfoutput=1
\usepackage[utf8]{inputenc}
\usepackage[english]{babel}
\usepackage[T1]{fontenc}
\usepackage{amsmath, amsfonts, amssymb, amsthm}
\usepackage{hyperref, url}
\usepackage[numbers,sort&compress]{natbib}

\usepackage{tikz}
\usepackage{lipsum}

\newtheorem{theorem}{Theorem}

\newtheorem{lemma}[theorem]{Lemma}

\newtheorem{remark}{Remark}

\newcommand{\ket}[1]{|#1\rangle}
\newcommand{\bra}[1]{\langle#1|}

\newcommand{\tr}{\text{\rm tr}}

\newcommand{\supp}{\text{\rm supp}}

\newcommand{\cC}{\mathcal{C}}

\newcommand{\cH}{\mathcal{H}}

\newcommand{\LHV}{\text{\rm{L}}}
\newcommand{\q}{\text{\rm{q}}}
\newcommand{\qs}{\text{\rm{qs}}}
\newcommand{\qa}{\text{\rm{qa}}}
\newcommand{\qc}{\text{\rm{qc}}}

\newcommand{\spn}{\text{\rm{span}}}
\newcommand{\rank}{\text{\rm{rank}}}
\newcommand{\Schrank}{\text{\rm{Sch-rank}}}
\newcommand{\sfA}{\mathsf{A}}
\newcommand{\sfB}{\mathsf{B}}
\newcommand{\sfE}{\mathsf{E}}
\newcommand{\sfF}{\mathsf{F}}

\begin{document}

\title{Separation of quantum, spatial quantum, and approximate quantum correlations}

\author{Salman Beigi}
\affiliation{QuOne Lab, Phanous Research and Innovation Centre, Tehran, Iran}
\orcid{0000-0003-3588-4662}
\email{salman@phanous.ir}

\maketitle

\begin{abstract}
Quantum nonlocal correlations are generated by implementation of local quantum measurements on spatially separated quantum subsystems. Depending on the underlying mathematical model, various notions of sets of quantum correlations can be defined. In this paper we prove separations of such sets of quantum correlations. In particular, we show that the set of bipartite quantum correlations with four binary measurements per party becomes strictly smaller once we restrict the local Hilbert spaces to be finite dimensional, \emph{i.e.}, $\cC_{\q}^{(4, 4, 2,2)} \neq \cC_{\qs}^{(4, 4, 2,2)}$.  We also prove non-closure of the set of bipartite quantum correlations with four ternary measurements per party, \emph{i.e.}, $\cC_{\qs}^{(4, 4, 3,3)} \neq \cC_{\qa}^{(4, 4, 3,3)}$.
\end{abstract}

\section{Introduction}

Nonlocality is one of the most fascinating features of quantum physics, stating that spatially separated parties can generate correlations that cannot be generated in the \emph{local hidden variable model}~\cite{EPR}. In a bipartite nonlocality scenario \`a la Bell~\cite{Bell64}, it is assumed that two parties, Alice and Bob can apply a measurement of their choice, which we denote by labels $s$ and $t$ respectively, on their respective subsystems. The measurement outcomes $a, b$, in the most general setting, are probabilistic. Thus we denote by $p(a, b|s, t)$ the probability of obtaining outputs $a, b$ by Alice and Bob when they apply measurements $s, t$ respectively. 
The set of all correlations $p(a, b|s, t)$ that can be represented in the local hidden variable model is denoted by $\cC_{\LHV}$, and is called the set of \emph{local correlations}. 
Bell nonlocality theorem~\cite{Bell64} states that there are correlations in the quantum world that do not belong to $\cC_\LHV$.

Correlations in the quantum theory are obtained as follows. Corresponding to Alice and Bob's local subsystems there are Hilbert spaces\footnote{All Hilbert spaces in this paper are assumed to be separable.} $\cH_{\mathsf A}$ and $\cH_{\mathsf B}$, and the state of their joint system is described by a unit vector $\ket\psi\in \cH_{\mathsf A}\otimes \cH_{\mathsf B}$. Alice and Bob's measurement settings $s, t$ correspond to projective measurements $\{P_s^{(a)}:~ a\}$ and $\{Q_t^{(b)}:~b\}$ respectively.\footnote{By enlarging the local subspaces $\cH_{\mathsf A}, \cH_{\mathsf B}$ and taking a purification we may assume that the shared state is pure. Moreover, by Naimark's dilation theorem we may assume that the local measurement operators are projective. We stick to these assumptions all over the paper.} Then the probability $p(a, b| s, t)$ of obtaining outputs $a, b$ is given by
\begin{align}\label{eq:quantum}
p(a, b| s, t)= \langle \psi | P_s^{(a)}\otimes Q_t^{(b)}\ket \psi. 
\end{align}
The set of quantum correlations $p(a, b| s, t)$ that can be written as above is denoted by $\cC_\qs$ and is called \emph{spatial quantum correlations}. We emphasize that for correlations in $\cC_\qs$ the local Hilbert spaces $\cH_{\mathsf A}, \cH_{\mathsf B}$ may have infinite dimensions. Then to distinguish the case of finite dimensional Hilbert spaces, we denote by $\cC_\q$ the set of correlations of the form~\eqref{eq:quantum} in which the local Hilbert spaces $\cH_{\mathsf A}, \cH_{\mathsf B}$ have finite dimensions.  By the definitions we have $\cC_\q \subseteq \cC_\qs$. We also note that by standard tricks local correlations can be realized in the quantum world as well. 
Indeed, local correlations have quantum representations in finite dimensions, so we have
$\cC_\LHV\subseteq \cC_\q$.

We may define yet two other sets of correlations related to the quantum theory. It is not hard to verify that $\cC_\LHV$ is closed. However, it is not clear that $\cC_\q$ and $\cC_\qs$ are closed. Thus we denote by $\cC_{\qa}$ the closure of $\cC_{\q}$ and call it the set of \emph{approximate quantum correlations}. As correlations in $\cC_\qs$ can be approximated by finite dimensional ones, we have $\cC_\qs\subseteq \cC_\qa$ and that $\cC_\qa$ is the closure of $\cC_\qs$ as well~\cite{ScholzWerner08}.

\begin{table}[t]
\renewcommand{\arraystretch}{1.3}
\begin{center}
  \begin{tabular}{ | c | l  |} 
     \hline 
    $\quad\cC_\LHV \quad$ & {\small local correlations} \\  \hline 
    $\cC_\q$ & {\small finite-dimensional quantum correlations} \\ \hline
    $\cC_{\qs}$ & {\small spacial (infinite-dimensional) quantum correlations} \\ \hline
    $\cC_{\qa}$ & {\small approximate quantum correlations (closure of $\cC_{\q}$ and $\cC_{\qs}$)} \\     \hline
    $\cC_{\qc}$ & {\small commuting quantum correlations} \\ 
    \hline 
  \end{tabular}
 \caption{\small Summary of correlation sets }
 \label{table:0}
\end{center}
  \end{table}

We would also like to briefly mention another generalization of quantum correlations motivated by quantum field theory. Here, instead of assuming that Alice and Bob have their individual subsystems and their associated Hilbert spaces, it is assumed that there is a single Hilbert space to which the shared state $\ket \psi$ belong. Then the measurement operators $P_s^{(a)}$ of Alice commute with measurement operators $Q_t^{(b)}$ so that the joint measurement $P_s^{(a)}Q_t^{(b)}$ makes sense.
The set of correlations obtained in this model is denoted by $\cC_\qc$ and is called the set of \emph{commuting quantum correlations}. Observe that thinking of $\cH_{\mathsf A}\otimes \cH_{\mathsf B}$ as a single Hilbert space, and replacing $P_s^{(a)}$ and $Q_t^{(b)}$ with $P_s^{(a)}\otimes I_{\mathsf B}$ and $I_{\mathsf A}\otimes Q_t^{(b)}$ of the previous picture, these measurement operators commute. Then $\cC_\qc$ contains quantum (finite or infinite dimensional) correlations. Moreover, it is known that $\cC_\qc$ is closed~\cite{Fritz12}. Therefore, $\cC_\qa\subseteq \cC_\qc$. 
Putting all these together the following hierarchy of sets of correlations is obtained:
\begin{align}\label{eq:hierarchy}
\cC_\LHV\subseteq\cC_\q\subseteq\cC_\qs\subseteq \cC_\qa\subseteq \cC_\qc.
\end{align}

To get a finer perspective of these sets let us assume that Alice's measurement label $s$ takes $n_A$ values, and there are $n_B$ values $t$ for Bob's measurement setting. Also let $m_A$ and $m_B$ be the number of possible values of $a$ and $b$, outputs of Alice and Bob respectively. That is, for instance, Alice chooses $s$ amongst $n_A$ possible choices and after measurement outputs $a$ among $m_A$ possible values. Then we denote the set of such correlations $p(a, b| s, t)$ in $\cC_\ast$ for $\ast\in\{ \LHV, \q, \qs, \qa , \qc\}$ by $\cC_\ast^{(n_A, n_B, m_A, m_B)}$ which indeed is a subset of $\mathbb R^{n_An_Bm_Am_B}$. Observe that if any of $n_A, n_B, m_A, m_B$ equals $1$, all the sets $\cC_\ast^{(n_A, n_B, m_A, m_B)}$ become trivial and collapse to the lower level $\cC_\LHV^{(n_A, n_B, m_A, m_B)}$. This is because when, \emph{e.g.}, $m_A=1$, Alice's output is fixed. Also if $n_A=1$, Bob knows Alice's measurement setting and can guess her output.   Thus in the following we always assume that $n_A, n_B, m_A, m_B\geq 2$.

Bell's theorem~\cite{Bell64} says that the first inclusion in~\eqref{eq:hierarchy} is strict, \emph{i.e.}, $\cC_\LHV\neq \cC_\q$. More precisely, by the example of the CHSH inequality $\cC_\LHV^{(2, 2, 2, 2)}\neq \cC_\q^{(2, 2, 2, 2)}$~\cite{CHSH}. Then (simply by ignoring some measurement settings and some output values) we also have $\cC_\LHV^{(n_A, n_B, m_A, m_B)}\neq \cC_\q^{(n_A, n_B, m_A, m_B)}$ for any $n_A, n_B, m_A, m_B\geq 2$. 

Separation of other inclusions in~\eqref{eq:hierarchy} is started by the work of Slofstra~\cite{Slofstra16} who showed that $\cC_{\qs}\neq \cC_{\qc}$. He then improved his result in~\cite{Slofstra17} and showed that $\cC_{\qs}\neq \cC_{\qa}$. This separation is proven in~\cite{Slofstra17} for the parameters $(n_A, n_B, m_A, m_B)=(185, 235, 8, 2)$. This result then improved in~\cite{Dykema+18} and~\cite{MusatRodram19} for the parameters  $(n_A, n_B, m_A, m_B)=(5, 5, 2, 2)$. Later, Coladangelo~\cite{Col19} gave a quite simple proof of this separation based on the ideas of \emph{self-testing} and \emph{entanglement embezzlement}~\cite{vanDamHayden03}.  Coladangelo's proof of $\cC_{\qs}\neq \cC_{\qc}$ although simpler, is for the parameters  $(n_A, n_B, m_A, m_B)=(5, 6, 3, 3)$ and does not improve the parameters of any of the previous results.

Separation of $\cC_\q$ and $\cC_\qs$ was first conjectured in~\cite{PalVertesi10} and then proved by Coladangelo and Stark in~\cite{ColStark18}. Their proof is quite elementary and is based on the idea of self-testing of the so called \emph{tilted CHSH inequality}~\cite{Acin+12}. The separation $\cC_\q\neq \cC_\qs$ in~\cite{ColStark18} is proven for the parameters $(n_A, n_B, m_A, m_B)=(4, 5, 3, 3)$.

We also add that the separation of $ \cC_\qa$ and $\cC_\qc$ was very recently established in the breakthrough work of Ji \emph{et al.}~\cite{Ji+20} who also showed the refutation of \emph{Connes' embedding conjecture.}

\paragraph{Our results:} In this paper we give new proofs for both the separations $\cC_\q\neq \cC_\qs$  and $\cC_\qs\neq \cC_\qa$. Here is our first result. 

\begin{theorem}\label{thm:C-q-C-qs}
$\cC_{\q}^{(4, 4, 2, 2)} \neq \cC_{\qs}^{(4,4, 2,2)}$.
\end{theorem}

This theorem is an improvement on the result of~\cite{ColStark18} that establishes this separation for $(n_A, n_B, m_A, m_B) = (4, 5, 3, 3)$. We also, building on this result, show that
$$\cC_{\q}^{(3, 4, 3, 2)} \neq \cC_{\qs}^{(3,4, 3,2)},$$
which again is an improvement on~\cite{ColStark18}, but is not comparable to Theorem~\ref{thm:C-q-C-qs}.

Similarly to~\cite{ColStark18} our proof of this theorem is based on the self-testing of tilted CHSH correlations. That is, we assume that the parties are playing two copies of the tilted CHSH game to self-test a certain entangled state. But these two tests are not independent, and there are certain correlations among their outputs. This forces a particular structure on the shared entangled state. Then by studying the ranks of measurement operators we conclude that they cannot be finite, so that $\cC_\q\neq \cC_\qs$. 

To emphasize on the main difference of our proof with that of~\cite{ColStark18}, we note that the shared entangled state $\ket\psi$ in~\cite{ColStark18} is taken to be $\sum_{i=0}^\infty \alpha^i \ket{i}\ket i$ up to a normalization factor for some $0<\alpha<1$. Then this state can be written as 
$$\sum_{i=0}^\infty \alpha^i \ket{i}\ket i=\sum_{j=0}^\infty \alpha^{2j}\Big( \ket{2j}\ket{2j} + \alpha\ket{2j+1}\ket{2j+1} \Big),$$
which up to local isometries equals $\big( \ket{00}+\alpha\ket{11}  \big)\otimes \ket{\psi'}$ for some auxiliary state $\ket{\psi'}$. So we may apply the tilted CHSH game to self-test the shared state $\big( \ket{00}+\alpha\ket{11}  \big)$. In another decomposition we may write 
$$\sum_{i=0}^\infty \alpha^i \ket{i}\ket i=   \sum_{j=1}^\infty \alpha^{2j-1}\Big( \ket{2j-1}\ket{2j-1} + \alpha\ket{2j}\ket{2j} \Big)+\ket0\ket 0,$$
which again up to local isometries equals $\big( \ket{00}+\alpha\ket{11}  \big)\otimes \ket{\psi'} \bigoplus \ket{22}$. Thus again we may apply the tilted CHSH game on the orthogonal subspace to $\ket 2$ to self-test $\big( \ket{00}+\alpha\ket{11}  \big)$. Considering these two tilted CHSH games and the non-trivial correlations between them, the main result of~\cite{ColStark18} is derived. Here our proof strategy departs from~\cite{ColStark18} by considering the starting shared state 
$$\sum_{i=-\infty}^{+\infty} \alpha^{|i|} \ket i\ket i.$$
This state can be written as $\big( \ket{00}+\alpha\ket{11}  \big)\otimes \ket{\psi'}$ up to local isometries in two different ways \emph{without any extra summand}. This would allow us to decrease the number of outputs from $m_A=m_B=3$ to $m_A=m_B=2$. However, after this modification, the argument of~\cite{ColStark18} based on the analysis of Schmidt coefficients of the shared state does not directly work. To overcome this difficulty, we further use self-testing properties of tilted CHSH correlations and build our arguments based on the ranks of the measurement operators. This gives us our first result $\cC_{\q}^{(4, 4, 2, 2)}\neq \cC_{\qs}^{(4, 4, 2, 2)}$.

 We note that $\cC_\q^{(2, n_B, 2, 2)}=\cC_\qs^{(2, n_B, 2, 2)}= \cC_\qa^{(2, n_B, 2, 2)}$ since any pair of projections (Alice's measurement operators) can be block-diagonalized with blocks of size at most 2. Thus, the first place that the separation of these sets is expected, and is conjectured in~\cite{PalVertesi10}, is for the parameters $(n_A, n_B, m_A, m_B)=(3, 3, 2, 2)$. Thus our result does not close the problem, yet it gets closer to what is conjectured.

Our second result is regarding the separation $\cC_\qs\neq \cC_\qa$.

\begin{theorem}\label{thm:C-qs-qa-0}
$\cC_\qs^{(4, 4, 3, 3)}\neq \cC_\qa^{(4, 4, 3, 3)}$.
\end{theorem}

This result is an improvement over~\cite{Col19} that shows this separation for larger parameters $(n_A, n_B, m_A, m_B)=(5, 6, 3, 3)$, but is not comparable to~\cite{Dykema+18, MusatRodram19}.  

Our proof idea of this result is similar to~\cite{Col19} and is based on self-testing and entanglement embezzlement. The only difference with~\cite{Col19} is that in order to self-test the maximally entangled state of Schmidt rank $3$, instead of the protocol of~\cite{Coladangelo18}, we use the protocol of~\cite{Sarkar+19} which allows to self-test this state with question sets of size $2$ and answer sets of size $3$. 

As mentioned above, our proof of $\cC_\qs^{(4, 4, 3, 3)}\neq \cC_\qa^{(4, 4, 3, 3)}$ is based on the recent result of~\cite{Sarkar+19}. Sarkar \emph{et al.}~in this paper show that any maximally entangled state of Schmidt rank $d$ can be self-tested via the Bell inequality of~\cite{SATWAP}.   
We describe this Bell inequality in details in the following section. We also present a proof of this self-testing property in Appendix~\ref{app:satwap} that is shorter than the original proof in~\cite{Sarkar+19} and may enlighten our understanding of the aforementioned Bell inequality.

In the next section we describe two self-testing protocols that will be used in the proofs of our main results. In the following two sections we will prove Theorem~\ref{thm:C-q-C-qs} and Theorem~\ref{thm:C-qs-qa-0}.

\section{Self-testing protocols}\label{sec:self-test}

Recall that the CHSH inequality states that for binary observables $A_s, B_t$, $s, t\in \{0,1\}$, with $\pm 1$ values we have
$$\langle A_0B_0\rangle + \langle A_0B_1\rangle + \langle A_1B_0\rangle-\langle A_1B_1\rangle\leq 2,$$
where $\langle AB\rangle$ is the expectation value of the observable $AB$. 
The CHSH inequality belongs to a family of such inequalities called the tilted CHSH inequalities~\cite{Acin+12}, which given a parameter $\beta\in [0,2)$ states that
$$\beta\langle A_0\rangle +\langle A_0B_0\rangle + \langle A_0B_1\rangle + \langle A_1B_0\rangle-\langle A_1B_1\rangle\leq 2+\beta.$$
Observe that for $\beta=0$ we recover the standard CHSH inequality. Moreover, similar to the CHSH inequality, the tilted CHSH inequality is violated in the quantum regime.  Indeed, there are quantum binary observables $A_s, B_t$, $s, t\in \{0,1\}$ and an entangled state $\ket \phi$ such that 
\begin{align}\label{eq:q-t-chsh}
\big\langle\phi\big|\big(\beta A_0\otimes I +A_0\otimes B_0 +  A_0\otimes B_1 +  A_1\otimes B_0- A_1\otimes B_1\big)\big|\phi\big\rangle=\sqrt{8+2\beta^2}.
\end{align}
Moreover, $\sqrt{8+2\beta^2}$ is the maximum value that can be gained by a quantum strategy. 

Let us describe the observables and the shared state that give~\eqref{eq:q-t-chsh}. Given $\beta\in [0, 2)$ define $\theta, \mu\in (0, \pi/4]$,  and $\alpha\in (0,1]$ by
$$\tan(\mu)=\sin (2\theta) = \sqrt{\frac{4- \beta^2}{4+\beta^2}}, \qquad \quad\alpha=\tan (\theta).$$
Let 
\begin{align}\label{eq:phi}
\ket{\Phi^{(\alpha)}} =\cos \theta\ket{00} + \sin\theta\ket{11} =\frac{1}{\sqrt{1+\alpha^2}} (\ket{00} + \alpha\ket{11}),
\end{align}
be a two-qubit state. Also let $\sigma_z$ and $\sigma_x$ be the Pauli operators and define
\begin{align}\label{eq:sigma-beta}
\sigma_{z}^{(\alpha)} := \cos \mu \sigma_z + \sin \mu \sigma_x, \qquad \sigma_{x}^{(\alpha)} = \cos \mu \sigma_z - \sin\mu \sigma_x.
\end{align}
Observe that $\sigma_{z}^{(\alpha)}, \sigma_{x}^{(\alpha)}$ are also binary observables. 
Next let 
$$A_0=\sigma_z, \qquad A_1=\sigma_x, \qquad B_0=\sigma_{z}^{(\alpha)}, \qquad B_1=\sigma_{x}^{(\alpha)}.$$
Then with the shared state $\ket{\Phi^{(\alpha)}}$ in~\eqref{eq:phi} and these four observables we obtain~\eqref{eq:q-t-chsh}.

Surprisingly, the above shared state and observables  represent essentially the unique strategy that achieves~\eqref{eq:q-t-chsh}. This is called the self-testing property of the tilted CHSH inequality and is  proven in~\cite{YN13, BampsPironio15}.

\begin{theorem}{\rm{(Self-testing of tilted CHSH~\cite{YN13, BampsPironio15})}}\label{thm:self-t-CHSH}
Let $\ket\psi_{\mathsf A\mathsf B}\in \cH_{\mathsf A}\otimes \cH_{\mathsf B}$ be a bipartite entangled state and $A_s, B_t$, $s, t\in \{0,1\}$ be binary observables (with $A_s^2=I,  B_t^2=I$) such that~\eqref{eq:q-t-chsh} holds.  Let $\widetilde \cH_{\mathsf A}= \supp(\tr_{\mathsf B}\ket\psi\bra\psi)$  and  $\widetilde\cH_{\mathsf B}= \supp(\tr_{\mathsf A}\ket\psi\bra\psi)$ where $\supp(X)$ is the support of the operator $X$. Then
$\widetilde\cH_{\mathsf A}$ is invariant under $A_s$, $s=0,1$ and $\widetilde \cH_{\mathsf B}$ is invariant under $B_t$, $t=0, 1$. Thus we can let $\widetilde A_s, \widetilde B_t$ be the restrictions of $A_s, B_t$ to these invariant subspaces respectively.
Then, there are auxiliary Hilbert spaces $\cH_{\mathsf A'}, \cH_{\mathsf B'}$, \emph{invertible isometries} $U:\widetilde\cH_{\mathsf A}\to \mathbb C^2\otimes \cH_{\mathsf A'} $ and $V:\widetilde\cH_{\mathsf B}\to \mathbb C^2 \otimes \cH_{\mathsf B'}$ and a state $\ket{\psi'}\in \cH_{\mathsf A'}\otimes \cH_{\mathsf B'}$ such that
\begin{itemize}
\item $U\otimes V\ket \psi = \ket{\Phi^{(\alpha)}}\otimes \ket{\psi'},$
\item $U \widetilde A_0U^\dagger = \sigma_z\otimes I_{\mathsf A'}$, $U \widetilde A_1U^\dagger = \sigma_x\otimes I_{\mathsf A'}$,
\item $V\widetilde B_0V^\dagger = \sigma_{z}^{(\alpha)}\otimes I_{\mathsf B'}$, $V\widetilde B_1V^\dagger = \sigma_{x}^{(\alpha)}\otimes I_{\mathsf B'}$.
\end{itemize}
Here $\ket{\Phi^{(\alpha)}}$ is given in~\eqref{eq:phi}, $\sigma_{z}^{(\alpha)}$ and $\sigma_{x}^{(\alpha)}$ are defined in~\eqref{eq:sigma-beta}, and $I_{\mathsf A'}, I_{\mathsf B'}$ are the identity operators acting on $\cH_{\mathsf A'} $ and $\cH_{\mathsf B'}$ respectively. 
\end{theorem}

There are some remarks in line.

\begin{remark}
In the statement of the theorem we assume that $U, V$ are not only isometries, but also invertible, \emph{i.e.}, they are onto. Indeed, we claim that for such isometries $U, V$, their images $U\widetilde\cH_{\mathsf A}$ and $V\widetilde \cH_{\mathsf B}$ are always of the form $\mathbb C^2 \otimes\cH'_{\mathsf A'}$ and $\mathbb C^2\otimes \cH'_{\mathsf B'}$ respectively, so that we may identify $\cH_{\mathsf A'}$ and $\cH_{\mathsf B'}$ in the statement of the theorem by $\cH'_{\mathsf A'}$ and $\cH'_{\mathsf B'}$ respectively. To prove this claim, we compute
\begin{align*}
U\tr_{\mathsf B}\big( \ket \psi \bra \psi\big)U^\dagger &= \tr_{\mathsf B}\big( U\otimes I\ket \psi \bra \psi U^\dagger\otimes I\big) \\
&=  \tr_{\mathbb C^2\otimes \mathsf B'}\big( U\otimes V\ket \psi \bra \psi U^\dagger\otimes V^\dagger\big)\\
& = \tr_{\mathbb C^2\otimes \mathsf B'}\big(\ket{\Phi^{(\alpha)}}\bra {\Phi^{(\alpha)}}\otimes \ket{\psi'}\bra{\psi'}\big)\\
& =  \cos^2\theta\big( \ket0\bra 0+ \alpha^2\ket 1\bra 1 \big)\otimes \tr_{\mathsf B'}\big(\ket{\psi'}\bra{\psi'}\big).
\end{align*}
where we used $V^\dagger V= I$, \emph{i.e.}, $V$ is an isometry, and the cyclic property of trace. This means that $U$ sends $\widetilde \cH_{\mathsf A'} = \supp \big(\tr_{\mathsf B} \ket \psi\bra \psi\big)$ to the support of the partial trace of $\ket{\Phi^{(\alpha)}}\otimes \ket{\psi'}$ which of course is of the form $\mathbb C^2 \otimes\cH'_{\mathsf A'}$. The proof for $V$ is similar. 
\end{remark}

\begin{remark}\label{remark:2}
Suppose that $\ket\psi$ in the statement of the theorem has finite Schmidt rank. Then by definition $\Schrank(\ket\psi) = \dim \widetilde \cH_{\mathsf A} = \dim \widetilde \cH_{\mathsf B}$, where $\Schrank(\ket\psi)$ is the Schmidt-rank of $\ket\psi$. On the other hand, 
$$\Schrank(\ket\psi)  = \Schrank(U\otimes V\ket\psi) = \Schrank(\ket {\Phi_\alpha})\cdot \Schrank(\ket{\psi'}).$$
Then using the fact that $U,V$ are invertible and by comparing dimensions, $\Schrank(\ket{\psi'}) =  \dim \cH_{\mathsf A'} = \dim \cH_{\mathsf B'}$ and $\tr_{\mathsf A'}(\ket{\psi'}\bra{\psi'}), \tr_{\mathsf B'}(\ket {\psi'}\bra{\psi'})$ are full-rank. 
\end{remark}

\begin{remark}
We note that the operators $A_s, B_t$, $s, t\in \{0,1\}$ may act on larger spaces than $\widetilde\cH_{\mathsf A}= \supp(\tr_{\mathsf B}\ket\psi\bra\psi)$  and  $\widetilde\cH_{\mathsf B}= \supp(\tr_{\mathsf A}\ket\psi\bra\psi)$, and the fact that the tilted CHSH inequality is maximally violated does not give any information about the action of these operators on the complementary subspaces. Because of this, in the statement of the theorem we had to replace $A_s, B_t$ with their restricted versions $\widetilde A_s, \widetilde B_t$.  
\end{remark}

By Theorem~\ref{thm:self-t-CHSH} any entangled state of Schmidt rank $2$ can be self-tested. It is then shown in~\cite{CGS17} that in fact any entangled state of finite Schmidt rank can be self-tested. In  the special case of the maximally entangled state of Schmidt rank $d$, it is shown that it can be self-tested via a correlation in $\cC_\q^{3, 4, d, d}$~\cite{Coladangelo18}.  However, one would expect to get self-testing protocols for maximally entangled states with smaller question and answer sets, particularly with $n_A=n_B=2$ and $m_A=m_B=d$. This problem was resolved recently in~\cite{Sarkar+19}. To state this result we need to develop some notation.

Binary observables are usually represented by unitary operators $A$ with eigenvalues $\pm 1$ (equivalently with $A^2=I$). Indeed, to a binary projective measurement $\{P_0, P_1\}$ one can associate the binary observable $A=P_0-P_1$. Conversely, having such $A$, the operators of the associated projective measurement are given by $P_0 =(I+A)/2 $ and $P_1=(I-A)/2$. We denote the projection on the $(-1)^{a}$ eigenspace of $A$ by $A^{(a)} = P_a= (I+(-1)^aA)/2$ for $a\in \{0,1\}$.
We can apply the same idea to represent \emph{$d$-valued measurements}, \emph{i.e.}, measurements with outcomes in $\{0,1, \dots, d-1\}$. 

Let $\{P_0, \dots, P_{d-1}\}$ be a $d$-valued projective measurement. Then define
$$A= \sum_{k=0}^{d-1} \omega^k P_k,$$
where $\omega=e^{2\pi \mathrm i/d}$ is a $d$-th root of unity. Observe that $A$ is a unitary and $A^d=I$. Conversely, eigen-decomposition of any unitary $A$ with $A^d=I$ corresponds to a $d$-valued projective measurement with measurement operators:
$$P_a= \frac{1}{d} \sum_{\ell=0}^{d-1} \omega^{- a\ell} A^\ell,\quad \qquad 0\leq a\leq d-1.$$
As in the binary case, we denote by $A^{(a)}=P_a$, for $0\leq a\leq d-1$, the projection on the eigenspace of $A$ with eigenvalue $\omega^a$. 
As a result, similar to the binary case, any Bell inequality with $d$-valued measurements, can be written in terms of unitary operators $A_i, B_j$ with $A_i^d=B_j^d=1$, and their powers. We call such unitaries \emph{$d$-valued observables}.

We now state our desired Bell inequality which following~\cite{Sarkar+19} we call the \emph{SATWAP Bell inequality}~\cite{SATWAP}. Let $A_0, A_1$ be Alice's observables with $A_0^d=A_1^d=I$ and $B_0, B_1$ be Bob's observables with $B^d_0=B^d_1=I$ as above. Then the SATWAP Bell operator is given by
\begin{align}\label{eq:SATWAP-op}
\mathcal{O}_d = \sum_{k=1}^{d-1} \Big(  r_k A_0^k\otimes B_0^{-k} + \bar{r}_k \omega^k A_0^k\otimes B_1^{-k} + \bar{r}_k A_1^k\otimes B_0^{-k}  + r_k A_1^k\otimes B_1^{-k}   \Big),
\end{align}
where as before $\omega=2^{2\pi \mathrm i/d}$,
$$r_k = \frac{1}{\sqrt 2} \omega^{\frac{2k-d}{8}} = \frac{1-\mathrm i}{2}\omega^{\frac{k}{4}},$$
and $\bar r_k$ is the complex conjugate of $r_k$ given by $\bar r_k = r_{d-k}$. It is known that the maximum of $\mathcal O_d$ in the local hidden variable model, \emph{i.e.}, in $\cC_{\LHV}^{(2, 2, d, d)}$, equals $\big[ 2\cot(\pi/4d)-\cot(3\pi/4d) -4\big]/2$. Also, for any entangled state $\ket\psi$ we have
$$\langle \psi | \mathcal O_d\ket\psi \leq 2(d-1).$$
Moreover, there is essentially a unique strategy to saturate the above inequality. 

\begin{theorem}{\rm{(Self-testing of SATWAP~\cite{Sarkar+19})}} \label{thm:satwap}
\begin{enumerate}
\item[{\rm{(i)}}] Let $\ket{\Phi_d} = \frac{1}{\sqrt d}\sum_{i=0}^{d-1} \ket i\ket i$ be the maximally entangled state of Schmidt rank $d$. Let 
$$Z=\sum_{i=0}^{d-1} \omega^i \ket i\bra i,$$
be the generalized $\sigma_z$-Pauli operator. Also let  $\ket J= \frac{1}{\sqrt d} \sum_{i=0}^{d-1}\ket i$.
Then letting $\ket\psi=\ket{\Phi_d}$ and
\begin{align}
&A_0=\omega^{-1/4} Z\Big(I-(1-\mathrm i)\ket J\bra J \Big),\label{eq:A1-satwap}\\
& A_1 =\omega^{1/4}Z\Big(I-(1+\mathrm i)\ket J\bra J \Big) ,\label{eq:A2-satwap}\\
&B_0= Z, \label{eq:B1-satwap}\\
& B_1 =\omega^{1/2} \Big(I-2\ket J\bra J\Big)Z,\label{eq:B2-satwap}
\end{align}
we obtain $\langle \Phi_d| \mathcal O_d\ket {\Phi_d} = 2(d-1)$.

\item[{\rm{(ii)}}] Conversely, suppose that $\ket\psi\in \cH_{\mathsf A}\otimes \cH_{\mathsf B}$ and $d$-valued observables $A_0, A_1$ and $B_0, B_1$ acting on $\cH_{\mathsf A}$ and $\cH_{\mathsf B}$ respectively, are such that 
$$\langle \psi| \mathcal O_d\ket \psi=2(d-1).$$
Then $\widetilde \cH_{\mathsf A}= \supp(\tr_{\mathsf B}\ket\psi\bra\psi)$  and  $\widetilde\cH_{\mathsf B}= \supp(\tr_{\mathsf A}\ket\psi\bra\psi)$  are invariant under  $A_s$ and $B_t$, $s, t\in \{0, 1\}$ respectively. Moreover, there are auxiliary Hilbert spaces $\cH_{\mathsf A'}$ and $\cH_{\mathsf B'}$ and \emph{invertible isometries} $U: \widetilde \cH_{\mathsf A}\to \mathbb C^d\otimes \cH_{\mathsf A'}$ and $V: \widetilde \cH_{\mathsf B}\to \mathbb C^d\otimes \cH_{\mathsf B'}$ such that 
\begin{itemize}
\item $U\otimes V \ket \psi = \ket{\Phi_d}\otimes \ket{\psi'}$ where $\ket{\psi'}\in \cH_{\mathsf A'}\otimes \cH_{\mathsf B'}$,
\item $U A_0 U^{\dagger}=  \omega^{-1/4}  Z\big(I-(1-\mathrm i) \ket J\bra J\big) \otimes I_{\mathsf B'},$

\item $U A_1 U^{\dagger}=  \omega^{1/4}  Z\big(I-(1+\mathrm i) \ket J\bra J\big) \otimes I_{\mathsf B'},$

\item $V\widetilde B_0V^\dagger= Z\otimes I_{\mathsf B'}$,
\item $V\widetilde B_1V^\dagger   = \omega^{1/2}\big (I - 2\ket J\bra J\big) Z\otimes I_{\mathsf B'},$
\end{itemize}
where $\widetilde A_s, \widetilde B_t$ are the restrictions of $A_s, B_t$ to the invariant subspaces $\widetilde\cH_{\mathsf A}$ and $\widetilde\cH_{\mathsf B}$ respectively.

\end{enumerate}
\end{theorem}

We notice that the unitary operators $A_s, B_t$ given in equations~\eqref{eq:A1-satwap}-\eqref{eq:B2-satwap} differ from those of~\cite{Sarkar+19}. However, this is not hard to verify that they are indeed  equivalent under local unitaries. In fact, this simple representation of the optimal strategy in part (i) of the theorem would help us to attain a simpler proof of the self-testing property in part (ii). We give a proof of this theorem in Appendix~\ref{app:satwap}.

\section{Quantum correlations in finite vs infinite dimensions}

In this section we prove our first main result stated in Theorem~\ref{thm:C-q-C-qs}. Let us recall that $\cC_{\q}^{(n_A, n_B, m_A, m_B)}$ is the set of correlations $p(a, b| s, t)$, with $s, t$ taking $n_A, n_B$ values and $a, b$ taking $m_A, m_B$ values respectively, which have representations of the form~\eqref{eq:quantum} with the local spaces $\cH_A, \cH_B$ being finite dimensional.   $\cC_{\qs}^{(n_A, n_B, m_A, m_B)}$ is defined similarly except that the local Hilbert spaces may be infinite dimensional (yet separable).

Let us start by introducing the shared entangled state to be used for generating the target nonlocal correlation. Let $\cH_{\mathsf Z}$ be the separable Hilbert space with orthonormal basis $\{\ket i_{\mathsf Z}:\, i\in \mathbb Z\}$.\footnote{We represent these basis vectors with the subscript $\mathsf Z$ to distinguish them from the computational basis vectors in $\mathbb C^2$.} 
Define $\ket\Psi\in \cH_{\mathsf Z}\otimes \cH_{\mathsf Z}$ by
\begin{align}\label{eq:psi-Z}
\ket\Psi = \frac{1}{\sqrt C}\sum_{i\in \mathbb Z} \alpha^{|i|} \ket i_{\mathsf Z}\ket i_{\mathsf Z},
\end{align}
where $0<\alpha<1$ is an arbitrary constant, $|i|$ is the absolute value of integer $i$ and $C$ is a normalization factor 
$$C = \sum_{i\in \mathbb Z} \alpha^{2|i|} = \frac{1+\alpha^2}{1-\alpha^2}\,.$$
$\ket \Psi$ can be written as $\frac{1}{\sqrt{1+\alpha^2}} \big(\ket{00} + \alpha\ket{11}\big)\otimes \ket{\psi'}$ up to local isometries in two different ways, one of which by \emph{pairing} the basis states as $\{\ket i_{\mathsf Z}:\, i\in \mathbb Z\} = \cup_{j\in \mathbb Z} \big \{  \ket {2j}_{\mathsf Z}, \ket{2j+1}_{\mathsf Z} \big \}  $ and the other one by pairing them as $\{\ket i_{\mathsf Z}:\, i\in \mathbb Z\} = \cup_{j\in \mathbb Z} \big \{  \ket {2j}_{\mathsf Z}, \ket{2j-1}_{\mathsf Z}  \big\}  $.
To make this more precise let us introduce two isometries 
$$W_{0}, W_2: \mathbb C^2\otimes \cH_{\mathsf Z} \to \cH_{\mathsf Z},$$
given by
\begin{align}\label{eq:W1}
W_0\,\ket{0}\ket j_{\mathsf Z} = 
\begin{cases} 
\ket{2j}_{\mathsf Z} & \quad j\geq 0,\\
\ket{{2j+1}}_{\mathsf Z} & \quad j<0,
\end{cases}
 \qquad \quad 
 W_0\ket1 \ket j_{\mathsf Z} = 
 \begin{cases}
 \big|{2j + 1}\big\rangle_{\mathsf Z} &\quad j\geq 0, \\
\ket{2j}_{\mathsf Z} & \quad j<0,
 \end{cases}
\end{align}
and
\begin{align}\label{eq:W2}
W_2\,\ket{0}\ket j_{\mathsf Z} = 
\begin{cases} 
\ket{2j-1}_{\mathsf Z} & \quad j> 0,\\
\ket{2j}_{\mathsf Z} & \quad j\leq 0,
\end{cases}
 \qquad \quad 
 W_2\ket1 \ket j_{\mathsf Z} = 
 \begin{cases}
 \big|{2j }\big\rangle_{\mathsf Z} &\quad j> 0, \\
\ket{2j-1}_{\mathsf Z} & \quad j\leq 0.
 \end{cases}
\end{align}
Indexing these operators by $0, 2$ (as opposed to $0, 1$) is for later convenience. 
Observe that 
$$W_r^\dagger \otimes W_r^\dagger \ket \Psi = \frac{1}{\sqrt{1+\alpha^2}} \big(\ket{00}+\alpha\ket {11}\big) \otimes \ket{\psi'_r}, \qquad \quad r\in \{0,2\},$$
for some $\ket{\psi'_0}, \ket{\psi'_2}\in \cH_{\mathsf Z}\otimes \cH_{\mathsf Z}$.
Thus using $\ket\Psi$ as the shared state, the tilted CHSH game can be played in two different ways. 
To this end, let us define the observables $A_s, B_t$, $s, t\in \{0,1,2,3\} $ according to Table~\ref{table:1ad}. Then $A_s, B_t$ for $s, t\in \{0,1\}$ generate the tilted CHSH correlation and $A_s, B_t$ for $s, t\in \{2, 3\}$ generate another copy of this correlation.

\begin{table}[t]
\renewcommand{\arraystretch}{2.1}
\begin{center}
  \begin{tabular}{ | c || c  |} 
     \hline 
    $\qquad A_0 = W_0   \big(\sigma_z\otimes I_{}\big) W_0^\dagger\qquad $ & $\qquad B_0= W_0   \big(\sigma_{z}^{(\alpha)}\otimes I_{}\big) W_0^\dagger\qquad$ \\ [.07in] \hline 
    $\qquad A_1 =W_0   \big(\sigma_x\otimes I_{}\big) W_0^\dagger\qquad $ & $\qquad B_1= W_0   \big(\sigma_{x}^{(\alpha)}\otimes I_{}\big) W_0^\dagger\qquad$ \\ [.07in] \hline\hline
    $\qquad A_2= W_2  \big (\sigma_z\otimes I_{}\big) W_2^\dagger\qquad$ & $\qquad B_2= W_2   \big(\sigma_{z}^{(\alpha)}\otimes I_{}\big) W_2^\dagger\qquad$ \\ [.07in]\hline
    $\qquad A_3= W_2   \big(\sigma_x\otimes I_{}\big) W_2^\dagger\qquad$ & $\qquad B_3= W_2   \big(\sigma_{x}^{(\alpha)}\otimes I_{}\big) W_2^\dagger\qquad$\\ [.07in]
    \hline
  \end{tabular}
 \caption{\small $A_s, B_t$, $s, t\in \{0, 1, 2, 3\}$ are binary observables acting on the Hilbert space $\cH_{\mathsf Z}$. Here $\sigma_x, \sigma_z$ are the Pauli matrices, and $\sigma_{x}^{(\alpha)}, \sigma_{z}^{(\alpha)}$ are defined in~\eqref{eq:sigma-beta}. }
 \label{table:1ad}
\end{center}
  \end{table}

As before, for a binary observable $M$ (\emph{i.e.}, an operator with $M^2=I$ and $M=M^\dagger$) let $M^{(a)}$, $a\in \{0,1\}$ be the orthogonal projection on the eigenspace with eigenvalue $(-1)^a$, so that $M=M^{(0)} - M^{(1)}$. Then the correlation generated by the shared state~\eqref{eq:psi-Z} and observables in Table~\ref{table:1ad} is given by
\begin{align}\label{eq:p-abc-t}
p(a, b| s, t) = \bra \Psi A_s^{(a)}\otimes B_t^{(b)}\ket \Psi. 
\end{align}
By the definitions we have $p\in \cC_{\qs}^{(4, 4, 2, 2)}$. 

\begin{theorem}\label{thm:main-1-1}
For any $\alpha\in (0,1)$ the nonlocal correlation $p$ defined above does not belong to $\cC_\q^{(4, 4, 2, 2)}$. In particular, $\cC_{\q}^{(4, 4, 2, 2)}\neq \cC_{\qs}^{(4, 4, 2, 2)}$.
\end{theorem}

\begin{proof}
Suppose that $p\in \cC_\q^{(4, 4, 2, 2)}$. Therefore, there are \emph{finite dimensional} Hilbert spaces $\cH_{\mathsf A}, \cH_{\mathsf B}$, a bipartite entangled state $\ket{\phi}\in \cH_{\mathsf A}\otimes \cH_{\mathsf B}$, and observables $P_s= P_s^{(0)} - P_s^{(1)}$, $Q_t = Q_t^{(0)} - Q_t^{(1)}$, $s, t\in \{0,1,2,3\}$ acting on $\cH_{\mathsf A}$ and $\cH_{\mathsf B}$ respectively, such that
$$p(a, b| s, t) = \bra \phi P_s^{(a)}\otimes Q_t^{(b)}\ket \phi,$$
where $p(a, b| s, t)$ is given by~\eqref{eq:p-abc-t}.
As mentioned before, $p_r(a, b| s, t)=p(a, b| s+r, t+r)$ for $s, t\in \{0,1\}$ form two copies of the tilted CHSH correlation for $r\in \{0, 2\}$. Then by Theorem~\ref{thm:self-t-CHSH} there are Hilbert spaces $\cH_{\mathsf A'_r}, \cH_{\mathsf B'_r}$, $r\in \{0,2\}$, and invertible isometries $U_r:\widetilde \cH_{\mathsf A}\to \mathbb C^2\otimes \cH_{\mathsf A'_r}$ and $V_r:\widetilde \cH_{\mathsf B}\to \mathbb C^2\otimes \cH_{\mathsf B'_r}$ such that 
\begin{align}\label{eq:phi-phi'-r}
U_r\otimes V_r\ket \phi = \ket{\Phi^{(\alpha)}}\otimes \ket{\phi'_r} =\frac{1}{\sqrt{1+\alpha^2}} \big(\ket{00}+\alpha\ket{11}\big)\otimes \ket{\phi'_r},
\end{align}
and
\begin{align*}
U_r P_rU_r^\dagger = \sigma_z\otimes I_{\mathsf A'_r}, \qquad\quad & U_rP_{r+1}U_r^\dagger = \sigma_x\otimes I_{\mathsf A'_r},\\
V_r Q_rV_r^\dagger = \sigma_{z}^{(\alpha)}\otimes I_{\mathsf B'_r}, \qquad\quad & V_rQ_{r+1}V_r^\dagger = \sigma_{x}^{(\alpha)}\otimes I_{\mathsf B'_r}.
\end{align*}
Here $\widetilde\cH_{\mathsf A}\subseteq \cH_{\mathsf A}$ and $\widetilde\cH_{\mathsf B}\subseteq \cH_{\mathsf B}$ are supports of $\tr_{\mathsf B} \ket \phi\bra\phi$ and $\tr_{\mathsf A}\ket \phi\bra \phi$ respectively, and $\ket{\phi'_r}\in \cH_{\mathsf A'_r}\otimes \cH_{\mathsf B'_r}$ is some bipartite state. 

As stated in Theorem~\ref{thm:self-t-CHSH}, the above equations show that $P_s$ and $Q_t$, $s, t\in \{0,1, 2, 3\}$, leave the subspaces $\widetilde \cH_{\mathsf A}$ and $\widetilde \cH_{\mathsf B}$ invariant, respectively. Thus by restricting everything to these subspaces, we may assume with no loss of generality that $\tr_{\mathsf A}\ket\phi\bra\phi$ and $\tr_{\mathsf B}\ket\phi\bra\phi$ are invertible and $\widetilde \cH_{\mathsf A}=\cH_{\mathsf A}$ and $\widetilde \cH_{\mathsf B}=\cH_{\mathsf B}$. 

Let us define
\begin{align}\label{eq:op-M}
M=\frac{1}{2 \cos \mu} \big(Q_2 + Q_{3}\big),
\end{align}
where as before $\mu\in (0, \pi/4]$ is given by $\tan(\mu)= \sqrt{\frac{4- \beta^2}{4+\beta^2}}$.
Then by the definitions of $\sigma_{z}^{(\alpha)}, \sigma_{x}^{(\alpha)}$ in~\eqref{eq:sigma-beta} we have
$$V_2 MV_2^\dagger = \sigma_z\otimes I_{\mathsf B'_2}.$$
Thus $M$ is an observable. Moreover, 
$$M^{(b)}=\frac 1 2\Big(I_{\mathsf B}+ (-1)^b\,\frac{Q_2 + Q_{3}}{2\cos \mu} \Big), \qquad \quad b\in \{0,1\},$$
is the projection on the eigenspace of $M$ with eigenvalue $(-1)^b$. We also have
\begin{align}\label{eq:rank-M-r-b}
\rank (M^{(b)}) =\rank (V_2M^{(b)}V_2^\dagger) = \rank(\ket b\bra b\otimes I_{\mathsf B'_2}) = \dim\cH_{\mathsf B'_2} .
\end{align}

Let us compute
\begin{align}
\bra \phi P_0^{(1)}\otimes M^{(0)}\ket\phi & =  \frac 12\bra \phi P_0^{(1)}\otimes I_{\mathsf B}\ket\phi + \frac{1}{4\cos \mu} \Big(\bra \phi P_0^{(1)}\otimes Q_2\ket\psi + \bra \phi P_0^1\otimes Q_3\ket\phi  \Big)\nonumber\\
& = \frac 12p(a=1| s=0) + \frac{1}{4\cos \mu} \sum_{t=2}^3 p(1, 0|0, t) - p(1, 1| 0, t)\nonumber\\
& =  \frac 12 \bra \Psi A_0^{(1)}\otimes I\ket\Psi + \frac{1}{4\cos \mu} \Big(\bra \Psi A_0^{(1)}\otimes B_2\ket\Psi + \bra \Psi A_0^{(1)}\otimes B_3\ket\Psi  \Big)\nonumber\\
& =   \bra \Psi A_0^{(1)}\otimes D^{(0)}\ket\Psi,\label{eq:p0-1-m-2-0-A}
\end{align}
where $\ket \Psi$ is given by~\eqref{eq:psi-Z}, the operator $A_s, B_t$, $s, t\in \{0, 1, 2, 3\}$ are given in Table~\ref{table:1ad}, and
$$D^{(0)} =\frac{1}{2}\Big( I_{\mathsf Z} + \frac{B_2+B_3}{2\cos \mu}\Big).$$
Moreover, $p(a=1| s=0)$ is the marginal distribution of the correlation $p(a=1, b | s=0, t)$ over the first output, that is independent of $t$.
Now by the definitions of $B_2, B_3$ and the isometry $W_2$ we have
$$D^{(0)} =\frac{1}{2}\Big( I + \frac{B_2+B_3}{2\cos \mu}\Big)= W_2 \big(\ket 0\bra 0\otimes I\big) W_2^\dagger = \ket 0\bra 0_{\mathsf Z} +\sum_{j=1}^\infty \ket{(2j-1)}\bra{(2j-1)}_{\mathsf Z}+\ket{-2j}\bra{-2j}_{\mathsf Z}.$$
We also note that
$$A_0^{(1)}= W_0\big(   \ket 1\bra 1\otimes I  \big)W_0^\dagger = \sum_{j=1}^\infty \ket{(2j-1)}\bra{(2j-1)}_{\mathsf Z}+\ket{-2j}\bra{-2j}_{\mathsf Z} =  D^{(0)}-\ket 0\bra 0_{\mathsf Z}. $$
Putting these together, and using the particular form of $\ket\Psi$ we find that
\begin{align*}
\bra \phi P_0^{(1)}\otimes M^{(0)}\ket\phi & =  \bra \Psi A_0^{(1)} \otimes A_0^{(1)}\ket \Psi + \bra \Psi \big(A_0^{(1)}\otimes \ket 0\bra 0_{\mathsf Z}\big)\,\ket \Psi\\
&  = \bra \Psi A_0^{(1)} \otimes A_0^{(1)}\ket \Psi\\
& = \bra \Psi A_0^{(1)} \otimes I_{\mathsf Z} \ket \Psi\\
& = p(a=1| s=0)\\
& =\bra \phi P_0^{(1)}\otimes I\ket\phi. 
\end{align*}
Then using $M^{(1)}= I-M^{(0)}$ we have
$\big\|\bra\phi P_0^{(1)}\otimes M^{(1)}\ket\phi\big\|=\bra \phi P_0^{(1)}\otimes M^{(1)}\ket\phi =0$. That is, $  P_0^{(1)}\otimes M^{(1)}\ket\phi=0$, or equivalently
$$P_0^{(1)}\otimes M^{(0)}\ket\phi = P_0^{(1)}\otimes I\ket\phi.$$
Next using $U_0P_0^{(1)}U_0^\dagger = \ket 1\bra 1\otimes I_{\mathsf A'_0}$ and~\eqref{eq:phi-phi'-r} we have
\begin{align*}
\frac{\alpha}{\sqrt{1+\alpha^2}} \ket{11}\otimes \ket{\phi'_0}&=\big(U_0P_0^{(1)}U_0^\dagger\otimes I\big )\big(U_0\otimes V_0\ket \phi\big)\\
& =  (U_0\otimes V_0) P_0^{(1)}\otimes I\ket\phi \\
& =  (U_0\otimes V_0) P_0^{(1)}\otimes M^{(0)}\ket\phi \\
& = \big(U_0P_0^{(1)}U_0^\dagger\otimes V_0 M^{(0)} V_0^\dagger \big )\big(U_0\otimes V_0\ket \phi\big)\\ & =  \frac{\alpha}{\sqrt{1+\alpha^2}}  \big(I\otimes V_0 M^{(0)} V_0^\dagger \big )  \ket{11}\otimes \ket{\phi'_0}.
\end{align*}
Therefore, taking the partial trace of both sides over the first subsystem we find that
$$\ket 1\bra 1\otimes \tr_{A'_0}\big(\ket {\phi'_0}\bra{\phi'_0}\big) = \big(V_0M^{(0)}V_0^\dagger\big)\Big(\ket 1\bra 1\otimes \tr_{A'_0}\big(\ket{\phi'_0}\bra{\phi'_0}\big)\Big) \big(V_0M^{(0)}V_0^\dagger\big).$$
As a result, based on Remark~\ref{remark:2} we have
\begin{align*}
\spn\{\ket 1\}\otimes \cH_{\mathsf B'_0} = \supp\Big(\ket 1\bra 1\otimes \tr_{A'_0}\big(\ket {\phi'_0}\bra{\phi'_0}\big)\Big) \subseteq \supp\big(V_0M^{(0)}V_0^\dagger\big).
\end{align*}
Next, by comparing the dimensions, using~\eqref{eq:rank-M-r-b} and $\dim \cH_{\mathsf B'_0} = \dim \cH_{\mathsf B'_2} = \frac 12 \dim \cH_{\mathsf B}$, we find that equality holds in the above inclusion. Equivalently, we obtain
$$V_0M^{(0)}V_0^\dagger = \ket 1\bra 1\otimes I_{\mathsf B'_0}.$$
Therefore,
\begin{align}
\bra \phi P_0^{(0)}\otimes M^{(0)} \ket \phi & = \bra \phi (U_0^\dagger \otimes V_0^\dagger) (U_0 P_0^{(0)}U_0^\dagger \otimes V_0M^{(0)}V_0^\dagger) (U_0\otimes V_0)\ket\phi\nonumber\\
& = \bra{\Phi_\alpha} \big(   \ket 0\bra 0 \otimes \ket 1\bra 1   \big) \ket{\Phi^{(\alpha)}}\nonumber\\
& = 0.\label{eq:p00-m2-0=0}
\end{align}
On the other hand, following similar computations as in~\eqref{eq:p0-1-m-2-0-A} we have
\begin{align*}
\bra \phi P_0^{(0)}\otimes M^{(0)} \ket \phi  & =\bra \Psi A_0^{(0)}\otimes D^{(0)}\ket\Psi\\
& \geq  \bra \Psi \big(\ket 0\bra 0_{\mathsf Z}\otimes \ket 0\bra 0_{\mathsf Z}\big)\ket\Psi\\
& = \frac 1{C^2},
\end{align*}
where we used the fact that the supports of both $A_0^{(0)}$ and $D^{(0)}$ contain $\ket 0_{\mathsf Z}\in \cH_{\mathsf Z}$. This is in contradiction with~\eqref{eq:p00-m2-0=0}. We are done.

\end{proof}

In the definition of the nonlocal correlation $p(a, b| s,t)$ in the statement of Theorem~\ref{thm:main-1-1} we could exchange the operators $A_2, A_3$ with operators $B_2, B_3$ in Table~\ref{table:1ad}. In that case we again get two copies of the tilted CHSH correlation (but with roles of the two players being  exchanged in one of them), and can follow similar steps as in the above proof to show that this new correlation does not belong to $\cC_\q^{(4, 4, 2, 2)}$. Indeed, we get a slightly simpler proof since in this case we would not need to introduce a new operator as in~\eqref{eq:op-M}. However, our particular choice of the correlation $p(a, b| s,t)$ would help us to get yet another separation of $\cC_\q$ and $\cC_{\qs}$.

\begin{theorem}
$\cC_{\q}^{(3, 4, 3, 2)} \neq \cC_\qs^{(3, 4, 3, 2)}$.
\end{theorem}

\begin{proof}
We first describe a correlation $q(a, b|s, t)$ in $\cC_\qs^{(3, 4, 3, 2)}$.  Let the local Hilbert spaces be $\cH_{\mathsf Z}$ as before and $\ket\Psi\in \cH_{\mathsf Z}\otimes \cH_{\mathsf Z}$ be as in~\eqref{eq:psi-Z}.  Next we need to introduce three projective measurements $\big\{\widetilde A_s^{(a)}:~ a\in \{0,1,2\}\big\}$, $s\in \{\ast,1, 3\}$,
 and four (binary) observables $B_t$, $t\in \{0,1,2,3\}$. 
Here for convenience the inputs of the first player are indexed by $s$ in $\{\ast, 1, 3\}$ instead of $\{0,1,2\}$. 
 We let $B_t$, $t\in \{0,1,2,3\}$ be as before given by Table~\ref{table:1ad}. We also let $\widetilde A_s$ for $s\in \{1, 3\}$ be binary observables with values in $\{0,1\}$, \emph{i.e.}, we assume that $\widetilde A_1^{(2)}=\widetilde A_3^{(2)} =0$, and let $\widetilde A_1^{(a)}=A_1^{(a)}, \widetilde A_3^{(a)}=A_3^{(a)}$ for $a\in \{0,1\}$ be given according to the binary observables $A_1, A_3$  in Table~\ref{table:1ad}.  Finally, we define $\big\{\widetilde A_\ast^{(a)}:~ a\in \{0,1,2\}\big\}$ by
\begin{align*}
\widetilde A_\ast^{(0)} & = W_0\big(\ket 0\bra 0 \otimes I\big)W_0^\dagger - \ket 0\bra 0_{\mathsf Z}, \\
\widetilde A_\ast^{(1)} & = W_0\big(\ket 1\bra 1 \otimes I\big)W_0^\dagger,\\
\widetilde A_\ast^{(2)} & = \ket 0\bra 0_{\mathsf Z}.
\end{align*}
These give the correlation
$$q(a, b| s, t) = \bra{\Psi} \widetilde A_s^{(a)} \otimes B_t^{(b)}\ket{\Psi},$$ 
which belongs to $\cC_{\qs}^{(3, 4, 3, 2)}$. 

We claim that $q$ does not belong to $\cC_{\q}^{(3, 4, 3, 2)}$. The main point behind the proof is that the binary observables $A_0, A_2$ given in Table~\ref{table:1ad} can be written in terms of $\widetilde A_\ast^{(a)}$, $a\in \{0,1,2\}$. In fact, using the definitions of the isometries $W_0, W_2$ we have
\begin{align}
A_0^{(0)} &= \widetilde A_\ast^{(0)} + \widetilde A_\ast^{(2)}, \qquad\qquad \quad   A_0^{(1)} = \widetilde A_\ast^{(1)},\label{eq:A-tilde-0} \\
A_2^{(0)} &= \widetilde A_\ast^{(1)} + \widetilde A_\ast^{(2)}, \qquad \qquad\quad   A_2^{(1)} = \widetilde A_\ast^{(0)}.\label{eq:A-tilde-1}
\end{align}
This means that the correlation $p\in \cC_{\qs}^{(4,4, 2, 2)}$ of the previous theorem can be written in terms of $q$. Indeed, by the definitions, 
$$p(a, b| s, t) = q(a, b| s, t), \qquad \quad s\in \{1, 3\}, t\in \{1, 2, 3, 4\}, a, b\in\{0,1\},$$
and if $s\in \{0, 2\}$, we have, for instance
\begin{align*}
p(0, b| 0, t) &= \bra {\Psi} A_0^{(0)} \otimes B_t^{(b)} \ket{\Psi}\\ 
& = \bra {\Psi} \widetilde A_\ast^{(0)} \otimes B_t^{(b)} \ket{\Psi}+\bra {\Psi} \widetilde A_\ast^{(2)} \otimes B_t^{(b)} \ket{\Psi}\\
& = q(0, b| \ast, t) + q(2, b| \ast, t).
\end{align*}
Therefore, if $q\in \cC_\q^{(3, 4, 3, 2)}$, $p$ belongs to $\cC_\q^{(4, 4, 2, 2)}$, contradicting Theorem~\ref{thm:main-1-1}. 

To make the above argument more precise, suppose that there are \emph{finite dimensional} Hilbert spaces $\cH_{\mathsf A}, \cH_{\mathsf B}$, $\ket\phi\in \cH_{\mathsf A}\otimes \cH_{\mathsf B}$, local measurements $\big\{\widetilde P_s^{(a)}:~a\in \{0,1,2\}\big \}$, $s\in \{\ast, 1, 3\}$, and binary observables $Q_t$, $t\in \{0,1,2,3\}$ such that 
$$q(a, b| s, t) = \bra \phi \widetilde P_s^{(a)}\otimes Q_t^{(b)}\ket \phi.$$
Since $\widetilde A_s^{(2)}=0$ for $s\in \{1, 3\}$, we have
$$\|P_1^{(2)}\otimes I_{\mathsf B}\ket\phi\|^2=\bra \phi \widetilde P_1^{(2)}\otimes I_{\mathsf B}\ket\phi=q(a=2| s) = \bra{\Psi} \widetilde A_s^{(2)} \otimes I\ket{\Psi}=0,$$  
and $P_1^{(2)}\otimes I_{\mathsf B}\ket\phi=0$. Thus we may think of  $\big\{\widetilde P_s^{(a)}:~a\in \{0,1,2\}\big \}$ for $s\in \{1, 3\}$ as binary measurements with values in $\{0,1\}$. More precisely,  for $s\in \{1, 3\}$, we may define the binary measurement $\big\{P_s^{(0)}, P_s^{(1)}\big\}$ by $P_s^{(0)} = \widetilde P_s^{(0)}+ \widetilde P_s^{(2)}$ and $P_s^{(1)}= \widetilde P_s^{(1)} $. Also we may define binary measurements $P_0, P_2$ based on~\eqref{eq:A-tilde-0} and~\eqref{eq:A-tilde-1}:
\begin{align*}
P_0^{(0)} &= \widetilde P_\ast^{(0)} + \widetilde P_\ast^{(2)}, \qquad\qquad \quad   P_0^{(1)} = \widetilde P_\ast^{(1)}, \\
P_2^{(0)} &= \widetilde P_\ast^{(1)} + \widetilde P_\ast^{(2)}, \qquad \qquad\quad   P_2^{(1)} = \widetilde P_\ast^{(0)}.
\end{align*}
Then it is not hard to verify that the correlation generated by the shared state $\ket{\phi}$ and binary observables $P_s, Q_t$, $s, t\in\{0,1,2,3\}$ in finite dimensions equals $p$. However, in Theorem~\ref{thm:main-1-1} we showed that $p$ does not belong to $\cC_\q^{(4, 4, 2, 2)}$. This is a contradiction.

\end{proof}

\section{Non-closure of the set of quantum  correlations}

In this section we prove our second main result (Theorem~\ref{thm:C-qs-qa-0}) that is the separation of $\cC_{\qs}^{(4, 4, 3, 3)}$ and $\cC_{\qa}^{(4, 4, 3, 3)}$. Recall that $\cC_{\qa}^{(4, 4, 3, 3)}$ is the closure of $\cC_{\qs}^{(4, 4, 3, 3)}$ as well as $\cC_{\q}^{(4, 4, 3, 3)}$. To prove this separation, we introduce a sequence of correlations $\big\{p_n(a, b| s, t):~ n\geq 1\big\}$ in  $\cC_\q^{(4, 4, 3, 3)}$ that converge to a correlation $p_\ast(a, b| s,t)$. By definition $p_\ast(a, b| s, t)$ belongs to  $\cC_{\qa}^{(4, 4, 3, 3)}$. Then we show that $p_\ast\notin \cC_\qs^{(4, 4, 3, 3)}$. This gives our main result.

In constructing $p_\ast(a, b | s, t)$ and proving that it does not belong to $\cC_{\qs}$ we follow similar steps as in~\cite{Col19}. The only difference is that instead of using the protocol of~\cite{Coladangelo18} to self-test the maximally entangled state of Schmidt rank $3$, we use Theorem~\ref{thm:satwap}. This would enable us to reduce the number of inputs in the target nonlocal correlation. 

Following~\cite{Col19} our construction of $p_\ast(a, b| s, t)$ is based on \emph{entanglement embezzlement}~\cite{vanDamHayden03, Leung+13}. Let $\ket{\psi}, \ket\phi$ be two bipartite states that are not equivalent up to local isometries, \emph{i.e.}, they have different \emph{multisets} of Schmidt coefficients. Then for any other bipartite state $\ket \chi$, the states $\ket \psi\otimes \ket \chi$ and $\ket\phi\otimes \ket\chi$ are still inequivalent up to local isometries. Nevertheless, by choosing an appropriate state $\ket\chi$ with \emph{increasing} Schmidt rank, $\ket \psi\otimes \ket \chi$ and $\ket\phi\otimes \ket\chi$ become \emph{approximately} equivalent up to local isometries with an arbitrarily small error.  This is called entanglement embezzlement, and besides self-testing this is the main ingredient of our separation theorem in this section. 

Let 
$$\ket{\Phi_3} = \frac 1{\sqrt 3}\big( \ket{00}+\ket{11} +\ket{22} \big),$$
be the maximally entangled state in $\mathbb C^3\otimes \mathbb C^3$. Also let 
$$\ket{\tau} = \frac{1}{\sqrt 2}\big(\ket{00} +\ket{22}\big) \in \mathbb C^3\otimes \mathbb C^3.$$
Then let  $\cH_{\sfE^n} = \cH_{\sfF^n} = \big(\mathbb C^3\big)^{\otimes n}$ and define $\ket{\chi_n}\in \cH_{\sfE^n}\otimes \cH_{\sfF^n}$ by
$$\ket{\chi_n}_{\sfE^n\sfF^n} = \frac{1}{\sqrt{C_n}} \sum_{j=1}^n \ket{00}^{\otimes j}\otimes \ket{\tau}^{\otimes (n-j)},$$
where $C_n\geq n$ is a normalization factor. 
Now let
\begin{align}\label{eq:psi_n}
\ket{\psi_n}= \ket{\Phi_3}\otimes \ket{\chi_n}.
\end{align}
This entangled state will be used as the shared state to define a nonlocal correlation $p_n(a, b|s, t)$ in $\cC_{\q}^{(4, 4, 3, 3)}$. Since $\ket{\Phi_3}$ is a part of this shared state, the two players can generate the SATWAP correlation by measuring this part. More precisely, following part (i) of Theorem~\ref{thm:satwap} we let
$A_0, A_1$ and $B_0, B_1$ be $3$-valued observables given by the first two rows of Table~\ref{table:2}.

\begin{table}[t]
\renewcommand{\arraystretch}{2.1}
\begin{center}
  \begin{tabular}{ | c || c  |} 
     \hline 
    $\qquad A_0 =  \omega^{-1/4} Z\Big(I-(1-\mathrm i)\ket J\bra J \Big)\otimes I_{\sfE^n}       \quad $ & $\quad B_0= Z\otimes I_{\sfF^n}\qquad$ \\ [.07in] \hline 
   $\qquad A_1 =\omega^{1/4}Z\Big(I-(1+\mathrm i)\ket J\bra J \Big)\otimes I_{\sfE^n} \quad $ & $\quad B_1= \omega^{1/2} \Big(I-2\ket J\bra J\Big)Z\otimes I_{\sfF^n}\quad$ \\ [.07in] \hline\hline
     $\quad A_2^{(0)} - A_2^{(1)}= \Gamma_n^{\dagger}  \big (\widetilde\sigma_z\otimes I_{\sfE^n}\big) \Gamma_n\qquad$ & $\quad B_2^{(0)}- B_2^{(1)}= \Lambda_n^{\dagger}   \big(\widetilde\sigma_{z}^{(\alpha)}\otimes I_{\sfF^n}\big) \Lambda_n\qquad$  \\ [.07in]\hline 
    $\quad A_3^{(0)} - A_3^{(1)}= \Gamma_n^\dagger   \big(\widetilde\sigma_x\otimes I_{\sfE^n}\big) \Gamma_n\qquad$      &
     $\quad B_3^{(0)} - B_3^{(1)}= \Lambda_n^{\dagger}   \big(\widetilde\sigma_{x}^{(\alpha)}\otimes I_{\sfF^n}\big) \Lambda_n\qquad$ \\ [.07in]
    \hline
  \end{tabular}
 \caption{\small $A_s, B_t$, $s, t\in \{0, 1, 2, 3\}$ are $3$-valued observables acting on the Hilbert spaces $\mathbb C^3\otimes\cH_{\sfE^n}$ and $\mathbb C^3\otimes \cH_{\sfF^n}$ respectively. Here $\omega= e^{2\pi \mathrm i/3}$, $\ket J= \frac{1}{\sqrt 3}\big(\ket 0+\ket 1+\ket 2\big)$ and $Z= \sum_{j=0}^2 \omega^j\ket j\bra j$.
 Moreover, $\sigma_x, \sigma_z$ are the Pauli matrices, and $\sigma_{x}^{(\alpha)}, \sigma_{z}^{(\alpha)}$ for $\alpha= \frac 1{\sqrt 2}$ are defined in~\eqref{eq:sigma-beta}. By $\widetilde \sigma_z:\mathbb C^3\to \mathbb C^3$ we mean the operator that acts on $\spn\{\ket 0, \ket 1\}$ as $\sigma_z$ and $\widetilde \sigma_z\ket 2=0$. Indeed, we have $\widetilde \sigma_z = \ket0\bra 0-\ket 1\bra 1$. $\widetilde \sigma_x, \widetilde \sigma_z^{(\alpha)}, \widetilde \sigma_x^{(\alpha)}$ are defined similarly. 
 We also let $A_s^{(2)}= \Gamma_n^\dagger\big(\ket 2\bra 2\otimes I_{\sfE^n}  \big)\Gamma_n $ and $ B_t^{(2)}= \Lambda_n^\dagger\big(\ket 2\bra 2\otimes I_{\sfF^n}  \big)\Lambda_n $ for $s, t\in \{2, 3\}$.}
 \label{table:2}
\end{center}
  \end{table}

To define $A_s, B_t$ for $s, t\in \{2, 3\}$ we need to introduce some notation. 
Let the unitary $\Gamma_n: \mathbb C^3\otimes \cH_{\sfE^n}\to  \mathbb C^3\otimes \cH_{\sfE^n}$ 
be given by its action on computational basis vectors as
$$\Gamma_n\ket{e_0}\ket{e_1\dots e_n} = \begin{cases}
\ket{e_1}\ket{e_2\dots e_n e_0} \qquad &\forall i,\,~e_i\in \{0,2\}, \\
\ket{e_0}\ket{e_1\dots e_n} \qquad &\text{ otherwise}.
\end{cases}$$
 We define $\Lambda_n: \mathbb C^3\otimes \cH_{\sfF^n}\to  \mathbb C^3\otimes \cH_{\sfF^n}$ similarly. Indeed, the action of $\Gamma_n$ and $\Lambda_n$ on the subspace spanned by $\{\ket 0, \ket 2\}^{\otimes (n+1)}$ is a cyclic shift, while they behave as identity on the orthogonal subspace.
Since $\ket{\tau}\otimes \ket{\chi_n}_{\sfE^n\sfF^n}$ belongs to the tensor product of the former subspace with itself, we have
\begin{align*}
  \Gamma_{n}\otimes \Lambda_{n} \,\ket{\tau}\otimes \ket{\chi_n}_{\sfE^n\sfF^n} &= \frac{1}{\sqrt{C_n}}   \Gamma_{n}\otimes \Lambda_{n} \Big(\sum_{j=1}^n \ket \tau\otimes \ket{00}^{\otimes j} \otimes \ket{\tau}^{\otimes (n-j)}\Big)\\
&  =    \frac{1}{\sqrt{C_n}}   \sum_{j=1}^n \ket{00}^{\otimes j} \otimes \ket{\tau}^{\otimes (n-j+1)}\\
& = \frac{1}{\sqrt{C_n}} \ket{00}\otimes \Big( \sum_{j=1}^n \ket{00}^{\otimes (j-1)} \otimes \ket{\tau}^{\otimes (n-j+1)}\Big)\\
& = \ket{00}\otimes \ket{\chi_n}+ \ket{00}\otimes \ket{\epsilon_n},
\end{align*}
where $\ket{\epsilon_n} = \frac{1}{\sqrt{C_n}} \big(\ket{00}^{\otimes n} - \ket{\tau}^{\otimes n}\big)$. We note that $C_n\geq n$, so $\|\ket{\epsilon_n}\| \leq  \sqrt{\frac{2}{n}}$. This means that $\{\ket {\chi_n}:~ n\geq 1\}$ is an \emph{embezzlement family} for local transformation of $\ket\tau$ to $\ket{00}$.

Now observe that 
$$\ket{\Phi_3} = \frac{1}{\sqrt 3} \big(\ket{00}+\ket{11} +\ket{22}\big) =  \frac{1}{\sqrt 3} \big(\sqrt 2\,\ket{\tau}+\ket{11} \big).$$
Therefore, 
\begin{align}
  \Gamma_{n}\otimes \Lambda_{n} \,\ket{\psi_n} &=   \Gamma_{n}\otimes \Lambda_{n} \,  \ket{\Phi_3} \otimes \ket{\chi_n}\nonumber\\
& = \frac{1}{\sqrt 3}\big(\sqrt 2\ket{00} +\ket{11}\big)\otimes \ket{\chi_n} + \sqrt{\frac 2 3} \ket{00}\otimes \ket{\epsilon_n}\nonumber\\
& = \frac{1}{\sqrt{1+\alpha^2}}\big(\ket{00} +\alpha\ket{11} \big)\otimes \ket{\chi_n}+ \sqrt{\frac 2 3} \ket{00}\otimes \ket{\epsilon_n},\label{eq:phi-tau-emb}
\end{align}
where $\alpha= \frac{1}{\sqrt 2}$.
Thus $\Gamma_{n}\otimes \Lambda_{n} \,\ket{\psi_n}$ approximately contains $\ket{\Phi^{(\alpha)}}=\frac{1}{\sqrt{1+\alpha^2}}(\ket{00} +\alpha\ket{11})$ in its first registers. Therefore, the two players may generate the tilted CHSH correlation by locally measuring it. 
To this end, we define the measurements $\big\{A_s^{(0)}, A_s^{(1)}, A_s^{(2)}\big\}$ and $\big\{B_t^{(0)}, B_t^{(1)}, B_t^{(2)}\big\}$ for $s, t\in \{2, 3\}$ according to the last two rows of Table~\ref{table:2}. We should explain that the observables in the tilted CHSH correlation are binary, yet here $A_s, B_t$ are $3$-valued observables. Nevertheless, as is clear from~\eqref{eq:phi-tau-emb}, the state $\Gamma_{n}\otimes \Lambda_{n} \ket{\psi_n} $ is locally orthogonal to both $\spn\{\ket 2\}\otimes \cH_{\sfE^n}$ and $\spn\{\ket 2\}\otimes \cH_{\sfF^n}$. Thus we may implement, \emph{e.g.}, Pauli measurements on the first registers of $\Gamma_{n}\otimes \Lambda_{n} \ket{\psi_n}$. This gives us $A_s^{(a)}, B_t^{(b)}$, for $s, t\in \{2, 3\}$ and $a, b\in \{0,1\}$ as in Table~\ref{table:2}. We also let 
$$A_s^{(2)}= \Gamma_n^\dagger\big(\ket 2\bra 2\otimes I_{\sfE^n}  \big)\Gamma_n, \qquad\quad B_t^{(2)}= \Lambda_n^\dagger\big(\ket 2\bra 2\otimes I_{\sfF^n}  \big)\Lambda_n ,\quad \quad s, t\in \{2, 3\}.$$
Then we have $\sum_a A_s^{(a)}=I\otimes I_{\sfE^n}$ and $\sum_b B_t^{(b)}=I\otimes I_{\sfF^n}$ for $s, t\in \{2, 3\}$, ensuring that $A_s, B_t$ are valid 3-valued observables. Nevertheless, we emphasize once again that since  $\Gamma_{n}\otimes \Lambda_{n} \ket{\psi_n}$ is locally orthogonal to both projections $\ket 2\bra 2\otimes I_{\sfE^n}$ and $\ket 2\bra 2\otimes I_{\sfF^n}$, the outcomes of these 3-valued measurements on $\ket{\psi_n}$ is never $a=2$ or $b=2$.

The shared state $\ket{\psi_n}$ given in~\eqref{eq:psi_n} and the measurements given in Table~\ref{table:2} give us a correlation $p_n(a, b| s, t)\in \cC_\q^{(4, 4, 3, 3)}$. 
Since $\cC_\q^{(4, 4, 3, 3)}$ is bounded, the sequence $\{p_n(a, b| s, t): n\geq 1\}$ has a limit $p_\ast(a, b| s, t)$. By definition $p_\ast$ belongs to $\cC_{\qa}^{(4, 4, 3, 3)}$.  In the following we identify some crucial properties of this limiting correlation. 

By construction $p_n(a, b| s, t)$ for $s, t\in \{0,1\}$ resembles the SATWAP correlation, so does $p_*(a, b|s, t)$. Also, by~\eqref{eq:phi-tau-emb}, the limiting correlation $p_*(a, b| s, t)$  for $s, t\in \{2, 3\}$ is the tilted CHSH correlation for $\alpha=1/\sqrt{2}$.  We also have
\begin{align*}
p_n(1, 1| 2, 0) &= \bra {\psi_n} A_2^{(1)} \otimes B_0^{(1)}\ket{\psi_n}\\
& = \bra{\psi_n} \Gamma_n^{\dagger} \big(\ket 1\bra 1\otimes I_{\sfE^n}\big)\Gamma_n\otimes \big(\ket 1\bra 1\otimes I_{\sfF^n}\big) \ket{\psi_n}\\
& = \bra{\psi_n}  \big(\ket 1\bra 1\otimes I_{\sfE^n}\big)\otimes \big(\ket 1\bra 1\otimes I_{\sfF^n}\big) \ket{\psi_n}\\
& = \frac{1}{3},
\end{align*}
where in the third line we use the fact that $\Gamma_n$ acts as identity on the subspace $\spn\{\ket 1\}\otimes \cH_{\sfE^n}$. We similarly have $p_n(a=1| s=2) =\bra{\psi_n} A_2^{(1)}\otimes (I\otimes I_{\sfF^n})\ket{\psi_n}=1/3$ and $p_n(b=1| t=0)=\bra{\psi_n} (I\otimes I_{\sfE^n})\otimes B_0^{(1)}\ket{\psi_n} = 1/3 $. Therefore,
$$p_\ast(1, 1| 2, 0)=p_\ast(a=1| s=2) = p_\ast(b=1| t=0) = 1/3.$$
To summarize, $p_*(a, b| s, t) \in \cC_{\qa}^{(4, 4, 3, 3)}$ has the following properties:
\begin{itemize}
\item[{\rm{(i)}}] $p_\ast(a, b| s, t)$ for $s, t\in \{0,1\}$ is the SATWAP correlation for $d=3$.
\item[{\rm{(ii)}}] $p_*(a, b| s, t)$  for $s, t\in \{2, 3\}$ is the tilted CHSH correlation for $\alpha=1/\sqrt{2}$
\item[{\rm{(iii)}}] $p_\ast(1, 1| 2, 0)=p_\ast(a=1| s=2) = p_\ast(b=1| t=0) = 1/3$.
\end{itemize}

\begin{theorem}\label{thm:main-2}
There is no $q(a, b| s, t) $ in $\cC_{\qs}^{(4, 4, 3, 3)}$ satisfying the above properties {\rm (i)}, {\rm (ii)} and  {\rm(iii)}. As a result, $\cC_{\qs}^{(4, 4, 3, 3)}\neq \cC_{\qa}^{(4, 4, 3, 3)}$.
\end{theorem}

\begin{proof}
Suppose that $q(a, b| s, t) \in \cC_{\qs}^{(4, 4, 3, 3)}$ satisfies the aforementioned properties. Suppose that $q(a, b| s, t)$ is obtained by  the shared state $\ket{\phi}\in \cH_{\sfA}\otimes \cH_{\sfB}$ and $3$-valued observables $P_s, Q_t$, $s, t\in \{0,1,2,3\}$:
$$q(a, b| s, t) = \bra\phi P_s^{(a)}\otimes Q_t^{(b)}\ket\phi.$$
Here, $P_s^{(a)}$ is the projection on the eigenspace of $P_s$ with eigenvalue $\omega^a$, where $\omega= e^{2\pi \mathrm i/3}$, and $Q_t^{(b)}$ is defined similarly. 

Using Theorem~\ref{thm:satwap}, property (i) implies that there are invertible isometries $U_0: \widetilde \cH_{\sfA}\to \mathbb C^3\otimes \cH_{\sfA'_0}$ and $V_0: \widetilde \cH_{\sfB}\to \mathbb C^3\otimes \cH_{\sfB'_0}$ such that 
\begin{align}\label{eq:local-unitary-0}
U_0\otimes V_0 \ket{\phi} = \ket{\Phi_3}\otimes \ket{\phi'_0},
\end{align}
and
$$V_0 \widetilde Q_0 V_0^\dagger = Z\otimes I_{\sfB'_0},$$
where $\widetilde \cH_{\sfA}$ and $\widetilde \cH_{\sfB}$ are supports of $\tr_{\sfB}\ket\phi\bra \phi$ and $\tr_{\sfA}\ket\phi\bra \phi$ respectively, $\widetilde Q_0$ is the restriction of $Q_0$ to $\widetilde \cH_{\sfB}$, and $\ket{\phi'_0}\in \cH_{\sfA'_0}\otimes \cH_{\sfB'_0}$. 

Similarly, using Theorem~\ref{thm:self-t-CHSH}, property (ii) implies that there are  invertible isometries $U_1: \widetilde \cH_{\sfA}\to \mathbb C^2\otimes \cH_{\sfA'_1}$ and $V_1: \widetilde \cH_{\sfB}\to \mathbb C^3\otimes \cH_{\sfB'_1}$ such that 
\begin{align}\label{eq:local-unitary-1}
U_1\otimes V_1 \ket{\phi} = \frac{1}{\sqrt{1+\alpha^2}} (\ket{00}+\alpha\ket{11})\otimes \ket{\phi'_1},
\end{align}
and
$$U_1 \widetilde P_2^{(1)} U_1^\dagger = \ket 1\bra 1\otimes I_{\sfA'_1},$$
where $\ket{\phi'_1}\in \cH_{\sfA'_1}\otimes \cH_{\sfB'_1}$ and $\widetilde P_2$ is the restriction of $P_2$ to $\widetilde \cH_{\sfA}$.

Next, property (iii) implies that 
$$\bra \phi \widetilde P_2^{(1)}\otimes \widetilde Q_0^{(1)}\ket\phi =\bra \phi \widetilde  P_2^{(1)}\otimes I_{\sfB}\ket\phi=\bra \phi I_{\sfA}\otimes \widetilde  Q_0^{(1)}\ket\phi=\frac 13.$$
This, in particular, gives
$$\big\|  \widetilde P_2^{(1)}\otimes \big( I_{\sfB} - \widetilde Q_0^{(1)}   \big) \ket\phi   \big \|^2 = \bra \phi \widetilde P_2^{(1)}\otimes (I_{\sfB}-\widetilde Q_0^{(1)})\ket\phi =0.$$
Thus, $\widetilde P_2^{(1)}\otimes  \widetilde Q_0^{(1)}   \ket\phi =\widetilde P_2^{(1)}\otimes  I_{\sfB}   \ket\phi$. By a similar argument we obtain $\widetilde P_2^{(1)}\otimes  I_{\sfB}   \ket\phi = I_{\sfA}\otimes \widetilde  Q_0^{(1)}   \ket\phi$. Then,
$$\widetilde P_2^{(1)}\otimes  I_{\sfB}   \ket\phi = I_{\sfA}\otimes  \widetilde Q_0^{(1)}   \ket\phi.$$
We compute
\begin{align*}
\frac{1}{\sqrt 3} \ket{11}\otimes \ket{\phi'_0} & =\big(I \otimes \ket 1\bra 1\otimes I_{\sfB'_0}\big) \ket{\Phi_3}\otimes \ket{\phi'_0}\\
& =\big(I \otimes   V_0\widetilde Q_0^{(1)}V_0^\dagger \big) \ket{\Phi_3}\otimes \ket{\phi'_0}\\
& =\big(U_0 \otimes   V_0\big) \big(I\otimes \widetilde Q_0^{(1)} \big) \big(U_0^\dagger \otimes V_0^\dagger\big) \ket{\Phi_3}\otimes \ket{\phi'_0}\\
& =\big(U_0 \otimes   V_0\big) \big(I\otimes \widetilde Q_0^{(1)} \big) \ket{\phi}\\
& =\big(U_0 \otimes   V_0\big) \big(\widetilde P_2^{(1)}\otimes I \big) \ket{\phi}.
\end{align*}
Therefore,
\begin{align*}
\frac{1}{\sqrt 3} \big(U_1 \otimes   V_1\big) \big(U_0^\dagger \otimes   V_0^\dagger\big)\, \ket{11}\otimes \ket{\phi'_0} & = \big(U_1 \otimes   V_1\big) \big(\widetilde P_2^{(1)}\otimes I \big)\ket{\phi}\\
& = \big(U_1\widetilde P_2^{(1)} U_1^\dagger \otimes I \big) \big(U_1 \otimes   V_1\big) \ket{\phi}\\
& = \frac{1}{\sqrt{1+\alpha^2}} \big( \ket 1\bra 1\otimes I_{\sfA'_1} \otimes I  \big) (\ket{00} + \alpha\ket{11})\otimes \ket{\phi'_1}\\
& = \frac{\alpha}{\sqrt{1+\alpha^2}} \ket{11})\otimes \ket{\phi'_1}\\
& = \frac{1}{\sqrt{3}} \ket{11})\otimes \ket{\phi'_1}.
\end{align*}
As a result, $\ket{11}\otimes \ket{\phi'_0}$ and $\ket{11}\otimes \ket{\phi'_1}$ are equivalent up to local isometries. Equivalently, letting $S_j$, $j\in \{0,1\}$, be the \emph{multiset} of the Schmidt coefficients of $\ket{\phi'_j}$ we find that 
\begin{align}\label{eq:S-0-S-1}
S_0=S_1.
\end{align}
On the other hand, comparing~\eqref{eq:local-unitary-0} and~\eqref{eq:local-unitary-1} we find that $\ket{\Phi_3}\otimes \ket{\phi'_0}$ and $\frac{1}{\sqrt{1+\alpha^2}} (\ket{00} + \alpha\ket{11}) \otimes \ket{\phi'_1}$ are equivalent up to local isometries, and have the same multisets of Schmidt coefficients. That is,
$$\frac{1}{\sqrt 3} S_0\cup \frac{1}{\sqrt 3} S_0\cup \frac{1}{\sqrt 3} S_0 = \sqrt{\frac{2}{3}}S_1\cup \frac{1}{\sqrt 3} S_1,$$
where by $xS$ we mean $xS=\{xs:~ s\in S\}$. Then taking the supremum of both sides and using~\eqref{eq:S-0-S-1} we find that
$$\sqrt{\frac 2 3} \,\sup S_1 =\frac {1}{\sqrt 3} \sup S_0 =  \frac{1}{\sqrt 3} S_1,$$
which is a contradiction since $\sup S_1\neq 0$. We are done.

\end{proof}

The main result of~\cite{Col19} is more general than a separation of $\cC_\qs$ and $\cC_{\qa}$. Indeed, in~\cite{Col19} a nonlocal game is introduced that has the following property: in order to win the game with probability $\epsilon$-close to optimal, the Schmidt-rank of the shared entangled state must be $2^{\Omega(\epsilon^{-1/8})}$. The main tool of~\cite{Col19} for proving this result is the so call \emph{stability} of the self-testing protocols under noise. We know that the self-testing property of the tilted CHSH correlation is stable under noise. However, this is not known for the SATWAP correlation. Thus, it is not clear if Theorem~\ref{thm:main-2} can be generalized to a result similar to that of~\cite{Col19}.




\newpage
\appendix

\section{Proof of Theorem~\ref{thm:satwap}}\label{app:satwap}

Let us recall that the SATWAP Bell operator is given by
\begin{align*}
\mathcal{O}_d = \sum_{k=1}^{d-1} \Big(  r_k A_0^k\otimes B_0^{-k} + \bar{r}_k \omega^k A_0^k\otimes B_1^{-k} + \bar{r}_k A_1^k\otimes B_0^{-k}  + r_k A_1^k\otimes B_1^{-k}   \Big),
\end{align*}
where $\omega=2^{2\pi \mathrm i/d}$,
$$r_k = \frac{1}{\sqrt 2} \omega^{\frac{2k-d}{8}} = \frac{1-\mathrm i}{2}\omega^{\frac{k}{4}},$$
and $\bar r_k$ is the complex conjugate of $r_k$.
Similar to~\cite{Sarkar+19} we start with a \emph{sum-of-squares} decomposition of $\mathcal O_d$.
For $1\leq k\leq d-1$ define 
$$C_{0, k} = r_k B_0^{-k} + \bar r_k \omega^k B_1^{-k}, \qquad C_{1,k} = \bar r_k B_0^{-k} + r_k B_1^{-k}.$$
Using $\bar r_k = r_{d-k}$ we find that  $C_{s,k}^\dagger = C_{s,d-k}$, $s\in \{0,1\}$. Then 
we have 
\begin{align*}
\mathcal O_d &= \sum_{k=1}^{d-1} A_0^k\otimes C_{0,k} + A_1^k\otimes C_{1,k}\\
& = \sum_{k=1}^{d-1} A_0^{d-k}\otimes C_{0,d-k} + A_1^{d-k}\otimes C_{1,d-k}\\
& = \sum_{k=1}^{d-1} A_0^{-k}\otimes C_{0,k}^\dagger + A_1^{-k}\otimes C_{1,k}^\dagger\\
& = \frac{1}{2} \sum_{k=1}^{d-1} A_0^k\otimes C_{0,k} +A_0^{-k}\otimes C_{0,k}^\dagger+ A_1^k\otimes C_{1,k}  + A_1^{-k}\otimes C_{1,k}^\dagger\\
& = \frac 1 2\sum_{k=1}^{d-1} \Big(2\, I\otimes I + I\otimes C_{0,k}^{\dagger}C_{0,k} + I\otimes C_{1,k}^{\dagger}C_{1,k}-  M_{0,k}M_{0,k}^\dagger - M_{1,k}M_{1,k}^\dagger \Big)   ,
\end{align*}
where
$$M_{s,k}= A_s^k\otimes I - I\otimes C_{s,k}^\dagger , \qquad \quad s\in \{0,1\}, ~1\leq k\leq d-1.$$
Next we have
\begin{align*}
\sum_{k=1}^{d-1} \Big( C_{0,k}^{\dagger}C_{0,k} + C_{1,k}^{\dagger}C_{1,k} \Big) &=\sum_{k=1}^{d-1} 2(|r_k|^2 + |\bar r_k|^2) I + (\bar r_k^2 \omega^k  + r_k^2)B_0^k B_1^{-k} +(r_k^2\omega^{-k} + \bar r_k^2)B_1^kB_0^{-k}\\
& = 2(d-1) I,
\end{align*}
where we used $\bar r_k^2 \omega^k  + r_k^2=r_k^2\omega^{-k} + \bar r_k^2=0$. Therefore, 
\begin{align*}
\mathcal O_d  & =  2(d-1) -  \frac 1 2\sum_{k=1}^{d-1} \Big(  M_{0,k}M_{0,k}^\dagger + M_{1,k} M_{1,k}^\dagger \Big). 
\end{align*}
This means that for any bipartite state $\ket \psi$ we have
\begin{align*}
\bra\psi\mathcal O_d \ket \psi &=  2(d-1) -  \frac 1 2\sum_{k=1}^{d-1} \Big( \big \|M_{0,k}^\dagger\ket\psi\big\|^2 + \big\| M_{1,k}^\dagger\ket\psi\big\|^2 \Big)\\
&\leq 2(d-1),
\end{align*}
and equality holds if and only if $M_{0,k}^\dagger\ket\psi=M_{1,k}^\dagger\ket\psi=0$. Equivalently, we have $\bra\psi\mathcal O_d \ket \psi =2(d-1)$ if and only if 
\begin{align}\label{eq:equality-case}
A_s^{-k}\otimes I\ket\psi = I \otimes C_{s,k}\ket \psi, \qquad\quad s\in \{0,1\},~ 1\leq k\leq d-1. 
\end{align}

This equation encourages us to prove the following lemma which will be used later. 

\begin{lemma}\label{lem:invariant}
Let $\ket\psi_{\mathsf A\mathsf B}\in\cH_{\mathsf A}\otimes \cH_{\mathsf B}$ be a bipartite state and $M, N$ be two operators such that 
$$M\otimes I \ket \psi = I \otimes N\ket \psi.$$
Then $\supp\big(\tr_{\mathsf B}\, \ket \psi\bra\psi\big)$ is invariant under $M$ and $\supp\big(\tr_{\mathsf A}\,\ket \psi\bra\psi\big)$ is invariant under $N$. 
\end{lemma}

\begin{proof} 
Observe that a vector $\ket v_{\mathsf A}$ belongs to $\supp\big(\tr_{\mathsf B}\, \ket \psi\bra\psi\big)$ if and only if there exists $\ket w_{\mathsf B}$ such that $\ket v_{\mathsf B}= I_{\mathsf A} \otimes \bra w_{\mathsf B} \cdot \ket \psi_{\mathsf A \mathsf B}$. For such a vector $\ket v$ we have
$$M\ket v = M_{\mathsf A}\otimes \bra w_{\mathsf B} \cdot \ket \psi_{\mathsf A\mathsf B} = I_{\mathsf A}\otimes \bra w_{\mathsf B} N_{\mathsf B}\ket \psi_{ \mathsf A \mathsf B}= I_{\mathsf A}\otimes \bra{w'}_{\mathsf B}\cdot \ket{\psi}_{\mathsf A \mathsf B},$$
where $\ket{w'} = N\ket w$. Thus $M\ket v$ itself belongs to $\supp\big(\tr_{\mathsf B}\, \ket \psi\bra\psi\big)$ and this subspace is invariant under $M$. By a similar argument $\supp\big(\tr_{\mathsf A}\, \ket \psi\bra\psi\big)$ is invariant under $N$. 

\end{proof}

We now give the proofs of parts (i) and (ii) of the theorem separately.

\paragraph{Proof of (i).} 
We first need to verify that $A_s, B_t$, $s, t\in \{0,1\}$ are valid $d$-valued observables.  The fact that they are unitary can easily be verified since they are multiplications of unitary matrices. Then we need to show that $A_s^d=B_t^d=I$. We obviously have $B_0^d=I$. To show this for the other operators we need to calculate $A_s^k$ and $B_1^k$, $s\in \{0, 1\}$.  We notice that $\bra J Z^k\ket J =0$ for any $1\leq k\leq d-1$. Therefore, for $1\leq k\leq d$ we have
\begin{align*}
B_1^k & = \omega^{k/2} \Big(\big(I-2\ket J\bra J\big)Z\Big)^k\\
& = \omega^{k/2}\Big( Z^k - 2\sum_{\ell=0}^{k-1} Z^{\ell} \ket J\bra J Z^{k-\ell}\Big) \\
& =  \omega^{k/2}\Big( I - 2\sum_{\ell=0}^{k-1} Z^{\ell} \ket J\bra J Z^{-\ell}\Big)Z^k.
\end{align*}
Indeed, if we let $\ket{J_\ell} = Z^{\ell}\ket J$, then $\{\ket{J_0}, \dots, \ket{J_{d-1}}\}$ forms an orthonormal basis and 
$$B_1^k = \omega^{k/2}\Big( I - 2\sum_{i=0}^{k-1}  \ket {J_\ell}\bra {J_\ell}\Big)Z^k.$$
Using this we find that $B_1^d=I$.
By similar calculations we obtain 
$$A_0^k=\omega^{-k/4} Z^k\Big(I-(1-\mathrm i) \sum_{\ell=0}^{k-1}\ket {J_{-\ell}}\bra {J_{-\ell}} \Big),\qquad  A_1^k =\omega^{k/4}Z^k\Big(I-(1+\mathrm i)\sum_{\ell=0}^{k-1}\ket {J_{-\ell}}\bra {J_{-\ell}} \Big), $$
from which $A_0^d=A_1^d=I$ follows.

Next, we need to show that~\eqref{eq:equality-case} holds. In this case, since the shared state is the maximally entangled state $\ket{\Phi_d}$, equation~\eqref{eq:equality-case} is equivalent to 
\begin{align*}
A_s^{-k} = C_{s,k}^T,
\end{align*}
where $T$ denotes transposition with respect to the computational basis. We compute
\begin{align*}
C_{0, k} &= r_k B_0^{-k} + \bar r_k \omega^k B_1^{-k}   \\
& =r_k B_0^{-1} \big(  I + r_k^{-1}\bar r_k \omega^k B_0^k B_1^{-k}    \big)\\
& = r_k B_0^{-k}\big( I+\mathrm i\omega^{k/2} B_0^{k}B_1^{-k}\big)\\
& = r_k Z^{-k}\Big( I+\mathrm i\Big( I - 2\sum_{i=0}^{k-1}  \ket {J_\ell}\bra {J_\ell}\Big) \Big)\\
& = \omega^{k/4} Z^{-k}\Big( I- (1+\mathrm i) \sum_{i=0}^{k-1}  \ket {J_\ell}\bra {J_\ell} \Big).
\end{align*}
Therefore,
\begin{align*}
C_{0, k}^T& = \omega^{k/4} \Big( I- (1+\mathrm i) \sum_{i=0}^{k-1}  \ket {J_{-\ell}}\bra {J_{-\ell}} \Big)Z^{-k} = A_0^{-k},
\end{align*}
where we used $\ket {J_\ell} \bra{J_\ell}^T = \ket{J_{-\ell}}\bra{J_{-\ell}}$. By a similar calculation $C_{1, k}^T = A_1^{-k}$ is verified.

\medskip
\paragraph{Proof of (ii).}
By assumptions we know that~\eqref{eq:equality-case}
holds. Then by Lemma~\ref{lem:invariant} the subspaces $\widetilde \cH_{\mathsf A}$ and $\widetilde \cH_{\mathsf B}$ are invariant under $A_s^{-k}$ and $C_{s,k}$ respectively. Therefore, $\widetilde \cH_{\mathsf A}$ is invariant under $A_s$,  and since $B_t$ can be written as a liner combination of $C_{0,d-1}, C_{1,d-1}$, we find that $\widetilde \cH_{\mathsf B}$ is invariant under $B_t$. Using these, by restricting everything to $\widetilde \cH_{\mathsf A}$ and $\widetilde \cH_{\mathsf B}$, we may assume with no loss of generality that the partial traces of $\ket\psi$ are invertible  and that $\widetilde \cH_{\mathsf A}=\cH_{\mathsf A}$ and $\widetilde \cH_{\mathsf B}=\cH_{\mathsf B}$.

Let $1\leq k, \ell\leq d-1$ be such that $k+\ell\leq d-1$. Using~\eqref{eq:equality-case} we obtain
\begin{align*}
I\otimes C_{s,k}C_{s,\ell}\ket \psi = A_s^{-\ell}\otimes C_{s,k}\ket \psi = A_s^{-(k+\ell)}\otimes I\ket\psi = I\otimes C_{s, k+\ell}\ket \psi. 
\end{align*}
Thus since partial traces of $\ket \psi$ are invertible, we have $C_{s,k}C_{s,\ell} = C_{s,k+\ell}.$
This means that
$$C_{s,k}= C_{s,1}^k.$$
Next by a similar argument (using $A_s^d=I$), we find that 
$$C_{s,k}^d=I.$$
Observe that $C_{s,k}^{\dagger} = C_{s,d-k}= C_{s,1}^{d-k}= C_{s,1}^{-k} = C_{s, k}^{-1}$. Thus $C_{s,k}$ is unitary. 

We compute 
$$C_{0, k} = r_k B_0^{-k} + \bar r_k \omega^k B_1^{-k} =  r_k B_0^{-k}\big( I+r_k^{-1}\bar r_k\omega^k B_0^{k}B_1^{-k}\big)=r_k B_0^{-k}\big( I+\mathrm i\omega^{k/2} B_0^{k}B_1^{-k}\big).$$
Therefore, since both $C_{0, k}$ and $\sqrt 2r_k B_0^{-k}=\omega^{\frac{2k-d}{8}}‌B_0^{-k}$ are unitary, 
$$L_k=(\sqrt 2r_k B_0^{-k})^{-1}C_{0,k}=\frac{1}{\sqrt 2} \big( I+\mathrm  i \omega^{k/2}B_0^{k}B_1^{-k}\big),$$
is unitary as well. Expanding $L_kL_k^\dagger=I$, we find that that $\omega^{k/2}B_0^{k}B_1^{-k} - \omega^{-k/2}B_1^{k}B_0^{-k}=0$. Equivalently, 
\begin{align}\label{eq:D-k-B-0-B-1}
D_k = \omega^{k/2}B_0^{k}B_1^{-k} = \omega^{-k/2}B_1^{k}B_0^{-k},
\end{align}
is self-adjoint. $D_k$ is also unitary since it is a multiplication of unitary operators. Thus, $D_k$ is a unitary with eigenvalues $\pm 1$. Therefore, 
\begin{align}\label{eq:D-k-def-P}
D_k = I-2P_k,
\end{align}
for some \emph{orthogonal projection} $P_k$, and we have
\begin{align}\label{eq:C-0-Pk}
C_{0,k} = r_k B_0^{-k}\big( I+\mathrm  i (I-2P_k)\big)= \omega^{k/4} B_0^{-k}\big(I-(1+\mathrm i)P_k\big).
\end{align}
By a similar calculation we have
\begin{align}\label{eq:C-1-Pk}
C_{1,k} = \bar r_k B_0^{-k}\big( I-\mathrm  i D_k\big)= \omega^{-k/4} B_0^{-k}\big(I-(1-\mathrm i)P_k\big).
\end{align}

We now use the fact that $C_{s, k+\ell}= C_{s,k}C_{s, \ell}$. By~\eqref{eq:C-0-Pk} this gives
$$B_0^{\ell}P_k B_0^{-\ell} +P_\ell -(1+\mathrm i) B_0^{\ell}P_kB_0^{-\ell}P_\ell = P_{k+\ell}.$$ 
Subtracting this equation from its adjoint, we find that $B_0^{\ell}P_kB_0^{-\ell}P_\ell+P_\ell B_0^{\ell}P_kB_0^{-\ell}=0$. This means that for any $\ket v\in \supp(P_\ell)$ we have  $2\bra v B_0^\ell P_k B_0^{-k}\ket v=0$. Then, using the fact that $B_0^\ell P_k B_0^{-k}$ is a projection, we find that $B_0^\ell P_k B_0^{-k}\ket v=0$ for all $\ket v\in \supp(P_\ell)$. As a result,
$$P_\ell B_0^{\ell}P_kB_0^{-\ell}=B_0^{\ell}P_kB_0^{-\ell}P_\ell=0,$$
and 
$$B_0^{\ell}P_k B_0^{-\ell} +P_\ell = P_{k+\ell}.$$
Then by a simple induction we arrive at 
\begin{align*}
P_k = \sum_{\ell=0}^{k-1} B_0^\ell P_1 B_0^{-\ell}.
\end{align*}
Indeed, $P_k$ is the summation of $k$ projections $B_0^\ell P_1 B_0^{-\ell}$, $0\leq \ell\leq k-1$ that are mutually orthogonal. 
Moreover, since $C_{0,1}^d=C_{0,1}C_{0,d-1}=I$ we have
$$I = \sum_{\ell=0}^{d-1} B_0^\ell P_1 B_0^{-\ell}.$$

Let us define
\begin{align}\label{eq:F-def-0}
F= \sum_{\ell=0}^{d-1} \omega^{-\ell} B_0^{\ell} P_1B_0^{-\ell}.
\end{align}
By the above equations, the orthogonal projections $B_0^{\ell} P_1B_0^{-\ell}$, $0\leq \ell\leq d-1$, form an eigen-decomposition of $F$. Thus, $F$ is a unitary operator and $F^d=I$.
Next, by a simple calculation we find that $B_0 F^k B_0^{-1} = \omega^{k} F^k$ and 
$$F^{-k}B_0 F^k = \omega^{k} B_0.$$
As a result, eigenspaces of $B_0$ are isomorphic. In fact, letting $\cH^{(k)}\subseteq \cH_{\mathsf B}$, $0\leq k\leq d-1$, be the eigenspace of $B_0$ with eigenvalue $\omega^k$, we have $\cH_{\mathsf B}=\bigoplus_k \cH^{(k)}$ and that 
\begin{align}\label{eq:cH-k-ell-F}
\cH^{(k+\ell)} = F^k \cH^{(\ell)}.
\end{align}
Let $\cH_{\mathsf B'} = \cH^{(0)}$ and define $V:\cH_{\mathsf B}\to \mathbb C^d\otimes \cH_{\mathsf B'}$ by
$$V\cH^{(k)} = \ket{k}\otimes F^{-k}\cH^{(k)}.$$
$V$ is well-defined since $\cH_{\mathsf B}=\bigoplus_k \cH^{(k)}$ and by~\eqref{eq:cH-k-ell-F} we have $F^{-k}\cH^{(k)}=\cH^{(0)} = \cH_{\mathsf B'}$. Moreover, since $F$ is unitary, $V$ is an invertible isometry. Now, for any $\ket{w_0}\in \cH_{\mathsf B'} = \cH^{(0)}$ by~\eqref{eq:cH-k-ell-F} there exists $\ket{w_k}\in \cH^{(k)}$ such that $\ket{w_0} = F^{-k} \ket{w_k}$. Then by the definition of $V$ we have
\begin{align*}
VB_0V^\dagger \ket k\otimes \ket {w_0} &= VB_0V^\dagger \ket k\otimes F^{-k} \ket {w_k}\\
& = VB_0\ket{w_k}\\
&= \omega^k V\ket{w_k}\\
& = \omega^k \ket k\otimes F^{-k}\ket{w_k}\\
& = \omega^k \ket k\otimes \ket{w_0}.
\end{align*}
This means that 
$$VB_0 V^{\dagger} = Z\otimes I_{\mathsf B'},$$
as desired. Next, taking $\ket{w_0}, \ket{w_k}$ as before, we have
\begin{align*}
VFV^\dagger \ket k\otimes \ket {w_0}&  = VFV^\dagger \ket k\otimes F^{-k}\ket {w_k} \\
& = VF \ket{w_k} \\
& = \ket{k+1} \otimes F^{-(k+1)}F\ket{w_k}\\
& = \ket{k+1} \otimes \ket{w_0},
\end{align*}
where in the third line we use $F\ket{w_k}\in \cH^{(k+1)}$. Therefore,
$$VFV^{\dagger} = X\otimes I_{\mathsf B'},$$
where $X:\mathbb C^d\to \mathbb C^d$ is defined by 
$$X\ket i=\ket{i+1 ~(\text{mod } d)}.$$

Now recall that by~\eqref{eq:F-def-0}, $P_1$ is the projection on the eigenspace of $F$ with eigenvalue $1$. Then using $VFV^{\dagger} = X\otimes I_{\mathsf B'}$ we find that
$$VP_1V^\dagger = \ket J\bra J\otimes I_{\mathsf B'},$$
where $\ket J = \frac 1{\sqrt d} \sum_{i=0}^{d-1} \ket i$.
Next, using this in~\eqref{eq:D-k-B-0-B-1} and~\eqref{eq:D-k-def-P} we obtain
\begin{align*}
VB_1V^{\dagger} = \omega^{1/2} V(D_1B_0)V^\dagger= \omega^{1/2} V\big(I-2P_1\big)B_0V^\dagger = \omega^{1/2}\big (I - 2\ket J\bra J\big) Z\otimes I_{\mathsf B'},
\end{align*}
as desired. Also using~\eqref{eq:C-0-Pk} and~\eqref{eq:C-1-Pk} we have
\begin{align}\label{eq:VCV-09}
VC_{0,1} V^\dagger = \omega^{1/4} Z^{-1}\big(I-(1+\mathrm i) \ket J\bra J\big) \otimes I_{\mathsf B'},
\end{align}
and 
\begin{align}\label{eq:VCV-01}
VC_{1,1} V^\dagger = \omega^{-1/4} Z^{-1}\big(I-(1-\mathrm i) \ket J\bra J\big) \otimes I_{\mathsf B'},
\end{align}

To characterize the state $\ket \psi$ and operators $A_0, A_1$, we once again use~\eqref{eq:equality-case}. 
We compute
\begin{align*} 
C_{s,k}\,\tr_{\mathsf A}\big(\ket \psi\bra \psi\big) \,  C_{s, k}^\dagger &= \tr_{\mathsf A}\big(   I_{\mathsf A}\otimes C_{s,k}\,\ket \psi\bra \psi \, I_{\mathsf A}\otimes C_{s, k}^\dagger    \big) \\
& = \tr_{\mathsf A}\big(   A_{s}^{-k}\otimes I_{\mathsf B}\,\ket \psi\bra \psi \, A_{s}^{k}\otimes I_{\mathsf B}    \big) \\
& = \tr_{\mathsf A}\big(   \ket \psi\bra \psi     \big).
\end{align*}
Thus $\tr_{\mathsf A}\big(\ket \psi\bra \psi\big)$ commutes with $C_{s,k}$, $s\in \{0,1\}$, $1\leq k\leq d-1$, and with any operator in the algebra generated by them. On the other hand, it is not hard to verify that the algebra generated by the operators in~\eqref{eq:VCV-09} and~\eqref{eq:VCV-01} equals the space of all operators of the form $Q\otimes I_{\mathsf B'}$.\footnote{Taking appropriate linear combinations, we find that both $Z^{-1}\otimes I_{\mathsf B'}$ and $\ket J\bra J\otimes I_{\mathsf B'}$ belong to this algebra. Then $Z^{k}\ket J\bra JZ^\ell\otimes I_{\mathsf B'}$ for all $0\leq k, \ell\leq d-1$ belong to this algebra, which span the space of operators of the form $Q\otimes I_{\mathsf B'}$.} 
As a result, $V\tr_{\mathsf A}\big(\ket \psi\bra \psi\big) V^\dagger$ commutes with any operator of the form $Q\otimes I_{\mathsf B'}$. This means that 
$$\tr_{\mathsf A}\big( (I_{\mathsf A}\otimes V) \ket \psi\bra\psi (I_{\mathsf A}\otimes V^\dagger )\big)=V\tr_{\mathsf A}\big(\ket \psi\bra \psi\big)V^\dagger = \frac 1 d I\otimes \rho'_{\mathsf B'},$$ 
for some density operator $\rho'$ acting on $\cH_{\mathsf B'}$. Therefore, there exists an invertible isometry $U: \cH_{\mathsf A}\to \mathbb C^d\otimes \cH_{\mathsf A'}$ such that\footnote{Here we use the fact that $\tr_{\mathsf B}\big(\ket \psi\bra\psi\big)$ is invertible.} 
$$U\otimes V\ket \psi=\ket{\Phi_d} \otimes \ket{\psi'}_{\mathsf A'\mathsf B'},$$
where $\ket{\psi'}_{\mathsf A'\mathsf B'}$ is a purification of $\rho'$.
Using this in~\eqref{eq:equality-case} we find that 
$$UA_s^{-1}U^{\dagger} = \big(V C_{s, 1}V^{\dagger}\big)^{T},$$
where $T$ is the transpose with respect to the computational basis of $\mathbb C^d$. Equivalently, we have
$$U A_0 U^{\dagger}= \big(V C_{0, 1}^{\dagger}V^{\dagger}\big)^{T}= \omega^{-1/4}  Z\big(I-(1-\mathrm i) \ket J\bra J\big) \otimes I_{\mathsf B'},$$
and
$$U A_1 U^{\dagger}= \big(V C_{1, 1}^{\dagger}V^{\dagger}\big)^{T}= \omega^{1/4}  Z\big(I-(1+\mathrm i) \ket J\bra J\big) \otimes I_{\mathsf B'}.$$
We are done.

\end{document}